\documentclass[10pt]{article}
\usepackage{amsfonts,color}
\usepackage{longtable}
\usepackage{caption}
\usepackage{multirow,multicol}
\usepackage{epsf, graphicx}
\usepackage{latexsym,amsfonts,amsbsy,amssymb}
\usepackage{amsmath,amsthm}
\usepackage{cite}
\usepackage{cleveref}
\usepackage{indentfirst}
\usepackage{bm}
\usepackage{amssymb,amsmath,amsthm,latexsym}
\usepackage{caption}
\usepackage{multirow,multicol}
\usepackage{epsf, graphicx}
\usepackage{latexsym,amsfonts,amsbsy,amssymb}
\usepackage{amsmath,amsthm}
\usepackage{cite}
\usepackage{cleveref}
\usepackage{indentfirst}
\usepackage{bm}
\usepackage [latin1]{inputenc}
\usepackage{enumitem}
\usepackage{tikz}
\usetikzlibrary{graphs}
\usetikzlibrary{positioning}
\usepackage{xcolor}
\usepackage{amsmath,amsthm,amssymb,arydshln}
\newtheorem{Theorem}{Theorem}[section]
\newtheorem{lem}[Theorem]{Lemma}
\newtheorem{Remark}[Theorem]{Remark}
\newtheorem{Definition}[Theorem]{Definition}
\newtheorem{Corollary}[Theorem]{Corollary}
\newtheorem{Construction}{Construction}[section]

\newtheorem{Example}[Theorem]{Example}
\numberwithin{equation}{section}

\begin{document}
	
	\title{Function-Correcting Codes for Symbol-Pair Read Channels\let\thefootnote\relax\footnotetext{E-Mail addresses: xiaqingfeng@mails.ccnu.edu.cn (Q. Xia), hwliu@ccnu.edu.cn (H. Liu), bocongchen@foxmail.com (B. Chen).}}
	\author{Qingfeng Xia$^1$,~ Hongwei Liu$^1$ and Bocong Chen$^2$}
	\date{\small $1.$ School of Mathematics and Statistics, Central China Normal University, Wuhan, 430079, China\\
$2.$ School of Future Technology, South China University of Technology, Guangzhou, 511442,  China\\}
	\maketitle
	{\noindent\small{\bf Abstract:} Function-correcting codes are a class of codes designed to protect the function evaluation of a message against errors whose key advantage is the reduced redundancy.
In this paper, we extend function-correcting codes from binary symmetric channels to  symbol-pair read channels.
We introduce irregular-pair-distance codes and connect them with function-correcting symbol-pair codes.
Using the connection, we derive general upper and lower bounds on the optimal redundancy of function-correcting symbol-pair codes.
For ease of evaluation, we simplify these bounds and employ the simplified bounds to specific functions including pair-locally binary functions, pair weight functions and pair weight distribution functions.
	
	\vspace{1ex}
	

\section{Introduction}
In standard communication systems, a sender desires to  transmit a message to a receiver via a noisy channel.
Classically, we assume that each part of the message is of equal importance to the receiver,
so the common goal is to construct an error-correcting code with a suitable decoder such that the  full message can be recovered correctly.
In reality, however, maybe only a certain attribute of the message, i.e.,
the result of evaluating a certain function on the message, is of particular interest to the receiver.
Of course, if the receiver is able to recover the whole message,
it is easy to  evaluate the function on the message to obtain the desired attribute; at any rate,
when the message is long and the function image is small, the solution is not efficient.
If we protect only the specific function value of interest, the redundancy will be reduced.
To meet such real need, Lenz, Bitar, Wachter-Zeh and Yaakobi \cite{LBWY} introduced a new paradigm, called function-correcting codes,
to put forward the study.

Since the key advantage of function-correcting codes is to reduce redundancy, in the  seminal paper \cite{LBWY},
the authors tried to use the smallest amount of redundancy that allows the recovery of the attribute;
in other words, they tried to find the optimal redundancy of function-correcting codes designed for a given function.
For this purpose, the notion of
irregular-distance codes was introduced and
it was showed that the optimal redundancy of a function-correcting code is given by the shortest length of an irregular-distance code which is the most important result of \cite{LBWY}.
Based on the connection, the other results of \cite{LBWY} lie in giving some general bounds on the optimal redundancy and employing these general results to specific functions.
The authors concluded the paper \cite{LBWY} by pointing out some further research directions, which
include the study of function-correcting codes under different channels.

In the era of information explosion, we need to store a huge amount of data,
so high density data storage devices are widely used in practice.
Motivated by the application of high-density data storage technologies, Cassuto and Blaum introduced the model of symbol-pair read channels
whose outputs are overlapping with pairs of symbols as shown in  Figure \ref{Figure1} below.
The model is based on the scenarios where, due to physical limitations, individual symbols cannot be read off the channel,
and therefore, each channel read contains contributions from two adjacent symbols.
It is worth noting that the advantage of the overlap between adjacent pair reads suggested by the model
over the natural method of partitioning the symbols to disjoint pairs.
It provides two observations of each symbol which offer better capabilities for error detection.
In the  classical  work \cite{CB}, Cassuto and Blaum introduced some fundamental notions about symbol-pair read channels,
established the relationship between the pair distance and the Hamming distance
and got the conclusion that a code $\mathcal{C}$ can correct $t$-pair errors if and only if its minimum pair distance $d_{p}\geq 2t+1$.
In addition, they provided some constructions and a decoding algorithm and derived the Sphere Packing bound and the Gilbert-Varshamov bound.
After that, the notions of symbol-pair read channel and symbol-pair  codes raised a lot of  scolars' interest, for example,
see \cite{CJKWY,W,EGY,YXW,CL,DGZZZ,KZL,LELP,LG,JML,ML,TL,THM} and the references therein.


Motivated by the aforementioned works, especially the
works of  \cite{LBWY,CB}, in this paper, we study
function-correcting codes
for symbol-pair read channels, which protect a certain attribute of the message transmitting through
the symbol-pair read channels, see Figure \ref{Figure2} below.
Analogous to that of \cite{LBWY}, the notion of
function-correcting symbol-pair codes is introduced.
Using function-correcting symbol-pair codes, we can reduce the redundancy and improve the efficiency of information storage and reading
under the condition that the receiver can recover the important attribute of the message.
For the large amounts of data in high-density storage devices, this may play a  significant role.

In this paper, we try to maximize the advantage of function-correcting symbol-pair codes,
so our central goal is to find upper and lower bounds on the optimal redundancy.
To this end,
we introduce the notion of irregular-pair-distance codes and connect them with function-correcting symbol-pair codes.
In the process, our main approach is to separate the message and redundancy vector of each codeword
which is more complicated than that in the  binary symmetric channel.
We derive upper and lower bounds on the optimal redundancy in terms of the shortest length on the irregular-pair-distance codes,
see Theorem \ref{theorem 3.1}.
Next, we employ these general results to specific functions including pair-locally binary functions, pair weight functions and pair weight distribution functions.
Let $t$ be the number of pair-errors. For pair-locally binary functions, we claim that the optimal redundancy is between $2t-2$ and $2t-1$.
And for pair weight functions, the optimal redundancy can be controlled between $\frac{20t^{3}-20t}{9(t+1)^{2}}$ and $\frac{4t-4}{1-2\sqrt{ln(2t-1)/(2t-1)}}$ under some conditions.
In addition, we find that the optimal redundancy of function-correcting symbol-pair codes for pair weight distribution functions is between $2t-2$ and $2t$.
Apart from these results, we provide some explicit constructions of function-correcting symbol-pair codes.

This paper is organized as follows.
In Section 2, we present the basic concepts and some fundamental results about irregular-distance codes, locally binary functions and symbol-pair codes. Next, we reveal the connection between function-correcting symbol-pair codes and irregular-pair-distance codes and give some general results in Section 3. In Sections 4, 5, 6, we apply those general results to pair-locally binary functions, pair weight functions and pair weight distribution functions respectively.
In Section 7, we compare the redundancy of classical symbol-pair codes with that of function-correcting symbol-pair codes.
In Section 8, we conclude our paper.

\begin{figure}[h]
  \centering
  \begin{tikzpicture}
  \node[rectangle,draw=white](A) at (0,0){$\boldsymbol{u}$};
  \node[rectangle,draw=black](B) at (2,0){Encoder};
  \node[rectangle,draw=black](C) at (6,0){Channel};
  \node[rectangle,draw=black](D) at (10,0){Decoder};
  \node[rectangle,draw=white](E) at (12,0){$\boldsymbol{u}$};
  \draw[->](0.5,0)--(1,0);
  \draw[->](3,0)--(5,0);
  \draw[->](7,0)--(9,0);
  \draw[->](11,0)--(11.5,0);
  \node[rectangle,draw=white](F) at (4,0.5){$\boldsymbol{c}$};
  \node[rectangle,draw=white](G) at (8,0.5){$\mathop{\boldsymbol{y}}\limits^{\leftrightarrow}$};
\end{tikzpicture}
  \caption{Classical Symbol-Pair Codes.}\label{Figure1}
\end{figure}
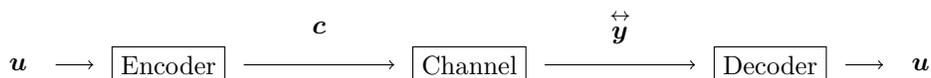

$\boldsymbol{u}\in\mathbb{Z}_{2}^{k}$ is a message to be read off. To guarantee recoverability of the message, we encodes the message to $\boldsymbol{c}\in\mathbb{Z}_{2}^{n}$. Through a symbol-pair read channel, the receiver gets a pair vector $\mathop{\boldsymbol{y}}\limits^{\leftrightarrow}\in(\mathbb{Z}_{2}\times\mathbb{Z}_{2})^{n}$ such that $d_{H}(\pi(\boldsymbol{c}),\mathop{y}\limits^{\leftrightarrow})\leq t$. According to $\mathop{\boldsymbol{y}}\limits^{\leftrightarrow}$, the receiver can correctly recover $\boldsymbol{u}$.

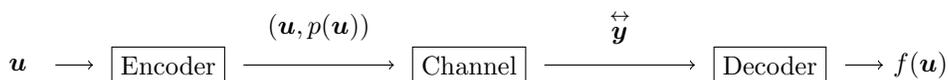
\begin{figure}[h]
  \centering
  \begin{tikzpicture}
  \node[rectangle,draw=white](A) at (0,0){$\boldsymbol{u}$};
  \node[rectangle,draw=black](B) at (2,0){Encoder};
  \node[rectangle,draw=black](C) at (6,0){Channel};
  \node[rectangle,draw=black](D) at (10,0){Decoder};
  \node[rectangle,draw=white](E) at (12,0){$f(\boldsymbol{u})$};
  \draw[->](0.5,0)--(1,0);
  \draw[->](3,0)--(5,0);
  \draw[->](7,0)--(9,0);
  \draw[->](11,0)--(11.5,0);
  \node[rectangle,draw=white](F) at (4,0.5){$(\boldsymbol{u},p(\boldsymbol{u}))$};
  \node[rectangle,draw=white](G) at (8,0.5){$\mathop{\boldsymbol{y}}\limits^{\leftrightarrow}$};
\end{tikzpicture}
  \caption{Function-Correcting Symbol-Pair Codes.}\label{Figure2}
\end{figure}

 $\boldsymbol{u}\in\mathbb{Z}_{2}^{k}$ is a message which features an attribute $f(\boldsymbol{u})$ that is of interest to the receiver. To guarantee recoverability of the attribute, we encodes the message to $(\boldsymbol{u},p(\boldsymbol{u}))\in\mathbb{Z}_{2}^{k+r}$. Through a symbol-pair read channel, the receiver gets a pair vector $\mathop{\boldsymbol{y}}\limits^{\leftrightarrow}\in(\mathbb{Z}_{2}\times\mathbb{Z}_{2})^{k+r}$ such that $d_{H}(\pi(Enc(\boldsymbol{u})),\mathop{y}\limits^{\leftrightarrow})\leq t$. With the knowledge of $f$ and $\mathop{\boldsymbol{y}}\limits^{\leftrightarrow}$, the receiver can correctly infer $f(\boldsymbol{u})$.

\section{Preliminaries}
In this section, we provide some basic knowledge about irregular distance codes,
locally binary functions and symbol-pair codes which will be used later. Throughout this paper,
we restrict our attention to binary codes, namely codes over the alphabet  $\mathbb{Z}_{2}$, even though some of the results can be easily generalized to larger alphabets.
\subsection{Irregular-Distance Codes and Locally Binary Functions}
We first review the notions  of irregular-distance codes and locally binary functions \cite{LBWY}.
As usual, let $\mathbb{N}_0$ denote the set of non-negative
integers. For a positive integer $M$, $[M]$ stands for the set $\{1,2,\ldots, M\}$ and $\mathbb{N}_{0}^{M\times M}$ stands for the set of
matrices having size $M\times M$ with entries being non-negative integers.
For a matrix $\boldsymbol{D}\in\mathbb{N}_{0}^{M\times M}$, we denote by $[\boldsymbol{D}]_{ij}$ the $(i,j)th$ entry of the matrix $\boldsymbol{D}$.
Given two vectors $\mathbf{x},\mathbf{y}\in \mathbb{Z}_2^n$, $d_H(\mathbf{x},\mathbf{y})$ denotes the Hamming distance between $\mathbf{x}$ and $\mathbf{y}$, i.e., $d_H(\mathbf{x},\mathbf{y})$ is equal to the number of coordinates where $\mathbf{x}$ and $\mathbf{y}$ differs.
\begin{Definition}
For a given matrix $\boldsymbol{D}\in\mathbb{N}_{0}^{M\times M}$, $\mathcal{P}=\{\boldsymbol{p}_{1},\boldsymbol{p}_{2},\ldots,\boldsymbol{p}_{M}\}\subseteq\mathbb{Z}_{2}^{r}$ is called a
$\mathbf{D}$-irregular-distance code ($\boldsymbol{D}$-code for short) if there exists an ordering of the codewords of $\mathcal{P}$ such that $d_H(\boldsymbol{p}_{i},\boldsymbol{p}_{j})\geq[\boldsymbol{D}]_{ij}$ for all $i,j\in[M]$.
In addition, $N(\boldsymbol{D})$ is defined to be the smallest integer $r$ such that there exists a $\boldsymbol{D}$-code of length $r$.
If $[\boldsymbol{D}]_{ij}=D$ for all $i\neq j$,  we write $N(\boldsymbol{D})$ as  $N(M,D)$.
\end{Definition}
Two upper bounds on $N(M,D)$ were reported in   \cite[Lemmas 3 and 4]{LBWY}; it seems that
both of them have small typos.
In \cite[Lemmas 3]{LBWY}, according to the definition of Hadamard codes given in \cite{H},  the order of a Hadamard matrix should be $2D$.
\begin{lem}(\cite[Lemma 3]{LBWY})\label{lemma 2.1}
  Let $D$ be a positive integer such that there exists a Hadamard matrix of order $2D$ and
Let $M$ be a positive integer satisfying $M\leq 4D$. Then, $N(M,D)\leq 2D.$
\end{lem}
 In \cite[Lemma 4]{LBWY}, it seems that $-2$ was missed  in the enumerator. The corrected form is given below.
\begin{lem}(\cite[Lemma 4]{LBWY})\label{lemma 2.2}
For any positive integers $M,D$ with $D\geq 10$ and $M\leq D^{2}$, we have
$$N(M,D)\leq\frac{2D-2}{1-2\sqrt{ln(D)/D}}.$$
\end{lem}

We now turn to  present the definition of locally binary functions, which was also  given in \cite{LBWY}.
Let $f:~\mathbb{Z}_2^k\rightarrow {\rm Im}(f)=\{f(\boldsymbol{u})\,|\,\boldsymbol{u}\in \mathbb{Z}_2^k\}$ be a function.
  The {\em function ball of  $f$ with radius $\rho$ around $\boldsymbol{u}\in\mathbb{Z}_{2}^{k}$} is defined by $$B_{f}(\boldsymbol{u},\rho)=\{f(\boldsymbol{u}^{\prime})\,|\,\boldsymbol{u}^{\prime}\in\mathbb{Z}_{2}^{k}, d_H(\boldsymbol{u},\boldsymbol{u}^{\prime})\leq\rho\}.$$

\begin{Definition}
  A function $f\,:\,\mathbb{Z}_{2}^{k}\rightarrow Im(f)$ is called a $\rho$-locally binary function
  if for all $\boldsymbol{u}\in\mathbb{Z}_{2}^{k}$, we have
  $|B_{f}(\boldsymbol{u},\rho)|\leq 2.$
\end{Definition}
\subsection{Symbol-Pair Codes}
In this subsection, we review some definitions and basic facts about symbol-pair codes.
Readers can refer to  \cite{CB,CJKWY,EGY} for more details.
Each vector in $\mathbb{Z}_{2}^{n}$ can be represented as a symbol-pair vector in $(\mathbb{Z}_{2}\times\mathbb{Z}_{2})^{n}$ as below.
\begin{Definition}
  Let $\boldsymbol{x}=(x_{0},\ldots,x_{n-1})$ be a vector in $\mathbb{Z}_{2}^{n}$. The symbol-pair read vector of $\boldsymbol{x}$ is defined as $$\pi(\boldsymbol{x})=\big((x_{0},x_{1}),(x_{1},x_{2}),\ldots,(x_{n-2},x_{n-1}),(x_{n-1},x_{0})\big).$$
\end{Definition}
We have seen that any vector $\boldsymbol{x}\in\mathbb{Z}_{2}^{n}$ has a pair representation $\pi(\boldsymbol{x})\in(\mathbb{Z}_{2}\times\mathbb{Z}_{2})^{n}$. Clearly, not all  vectors in $(\mathbb{Z}_{2}\times\mathbb{Z}_{2})^{n}$
can be written in the form $\pi(\boldsymbol{x})$ for some $\boldsymbol{x}\in \mathbb{Z}_2^n$. Pair vectors
that can be expressed in the form $\pi(\boldsymbol{x})$ are said to be  {\em consistent}.
In the   symbol-pair read channel, the number of pair-errors is bounded by an integer $t$.
The following gives the definition of $t$-pair errors, see \cite{CB}.
\begin{Definition}
Let $\boldsymbol{x}=(x_{0},\ldots,x_{n-1})$ be a vector in $\mathbb{Z}_{2}^{n}$.
A  vector
$$\mathop{\boldsymbol{y}}\limits^{\leftrightarrow}=((\lhd y_{0},\rhd y_{0}),\ldots,\\(\lhd y_{n-1},\rhd y_{n-1}))\in(\mathbb{Z}_{2}\times\mathbb{Z}_{2})^{n}$$
is the result of $t$-pair errors from $\boldsymbol{x}$ if $d_{H}(\pi(\boldsymbol{x}),\mathop{\boldsymbol{y}}\limits^{\leftrightarrow})\leq t$.
In the expression of $\mathop{\boldsymbol{y}}\limits^{\leftrightarrow}$, $\lhd$ and $\rhd$ designate the left and right symbols respectively.
\end{Definition}
After defining the $t$-pair errors, the next natural step is to prove the necessary and sufficient condition on the code for achieving correctability of $t$-pair errors. For the purpose, we need to define the pair distance.
\begin{Definition}(\cite{CB})
  For $\boldsymbol{x},\boldsymbol{y}\in\mathbb{Z}_{2}^{n}$, the pair distance between $\boldsymbol{x}$ and $\boldsymbol{y}$ is defined as $$d_{p}(\boldsymbol{x},\boldsymbol{y})=d_{H}(\pi(\boldsymbol{x}),\pi\big(\boldsymbol{x})\big).$$
\end{Definition}
The pair-distance is related to the Hamming distance in the following lemma.
\begin{lem}(\cite{CB})\label{lemma 2.3}
  For $\boldsymbol{x}$, $\boldsymbol{y}$ in $\mathbb{Z}_{2}^{n}$,
  we have $$d_{H}(\boldsymbol{x},\boldsymbol{y})+1\leq d_{p}(\boldsymbol{x},\boldsymbol{y})\leq 2d_{H}(\boldsymbol{x},\boldsymbol{y}).$$In the extreme cases in which $d_{H}(\boldsymbol{x},\boldsymbol{y})$ equals $0$ or $n$, clearly $d_{p}(\boldsymbol{x},\boldsymbol{y})=d_{H}(\boldsymbol{x},\boldsymbol{y}).$
\end{lem}
Let $\mathcal{C}\subset\mathbb{Z}_{2}^{n}$ be a code and let $$d_{p}=\min\limits_{\boldsymbol{x}\in\mathcal{C},\boldsymbol{y}\in\mathcal{C},\boldsymbol{x}\neq\boldsymbol{y}}d_{p}(\boldsymbol{x},\boldsymbol{y})$$ be the minimum pair-distance of $\mathcal{C}$. The necessary and sufficient condition for correctability of $t$-pair errors is given as follows.
\begin{lem}(\cite{CB})\label{lemma 2.4}
  A code $\mathcal{C}$ can correct $t$-pair errors if and only if $d_{p}\geq 2t+1$.
\end{lem}
We recall two bounds on the size of symbol-pair codes . Firstly, we represent the Sphere Packing bound, which was established in \cite{CB}.
\begin{lem}
  If $\mathcal{C}\subset\mathbb{Z}_{2}^{n}$ is a code with $M$ codewords that correct all $t$-pair errors, then $MB(n,t)\leq 2^{n},$ where $B(n,t)$ is the number of words with pair-distance $t$ or less from $\boldsymbol{x}\in\mathbb{Z}_{2}^{n}$.
\end{lem}
Next we review the computational formula of $B(n,t)$. For integers $n>l\geq L$, let $D(n,l,L)$ be the number of sequences of length $n$ and
Hamming weight $l$ that have $L$ runs. It follows from \cite{EGY} that
 $$D(n,l,L)=\frac{n}{L}\left(\begin{array}{c}
                              l-1 \\
                              L-1
                            \end{array}\right)\left(\begin{array}{c}
                              n-l-1 \\
                              L-1
                            \end{array}\right).$$
We denote by $S(n,i)$ the number of words with pair-distance $i$ from $\boldsymbol{x}\in\mathbb{Z}_{2}^{n}$. Then $$S(n,i)=\sum_{l=\lceil\frac{i}{2}\rceil}^{i-1}D(n,l,i-l)$$ and $$B(n,t)=1+\sum_{i=1}^{t}S(n,i).$$

We end this section by restating  the Singleton bound for symbol-pair codes \cite{CJKWY}.
\begin{lem}
  If $\mathcal{C}\subset\mathbb{Z}_{2}^{n}$ is a code with the minimum pair-distance $d_{p}$, the number of codewords $M$ satisfies
  $M\leq 2^{n-d_{p}+2}.$
\end{lem}

\section{General Results on the Optimal Redundancy}
In this section, we mainly establish the general results on function-correcting symbol-pair codes (FCSPCs).
For this purpose, we need to describe a relation between FCSPCs and irregular-pair-distance codes.
\subsection{The Connection between FCSPCs and Irregular-Pair-Distance Codes}
Let $\boldsymbol{u}\in\mathbb{Z}_{2}^{k}$ be a binary message and let
$f\,:\,\mathbb{Z}_{2}^{k}\rightarrow Im(f)=\{f(\boldsymbol{u})\,|\,\boldsymbol{u}\in\mathbb{Z}_{2}^{k}\}$
be a function computed on $\boldsymbol{u}$ with  $E=|Im(f)|\leq 2^{k}$.
The message is encoded via the encoding function $$Enc\,:\,\mathbb{Z}_{2}^{k}\rightarrow\mathbb{Z}_{2}^{k+r}\,,\,~~~Enc(\boldsymbol{u})=(\boldsymbol{u},\boldsymbol{p(u)}),$$
where $\boldsymbol{p(u)}\in\mathbb{Z}_{2}^{r}$ is
the redundancy vector and $r$ is the redundancy. The resulting codeword $Enc(\boldsymbol{u})$ is transmitted over a symbol-pair read channel, resulting in $\mathop{\boldsymbol{y}}\limits^{\leftrightarrow}\in(\mathbb{Z}_{2}\times\mathbb{Z}_{2})^{k+r}$ with $d_{H}(\pi(Enc(\boldsymbol{u})),\mathop{\boldsymbol{y}}\limits^{\leftrightarrow}) \leq t$.
For symbol-pair read channels, we introduce the following definition.
\begin{Definition}
 An encoding function $Enc\,:\,\mathbb{Z}_{2}^{k}\rightarrow\mathbb{Z}_{2}^{k+r}$, $Enc(\boldsymbol{u})=(\boldsymbol{u},\boldsymbol{p(u)})$ defines a function-correcting symbol-pair code (FCSPC for short) for the function $f\,:\,\mathbb{Z}_{2}^{k}\rightarrow Im(f)$ if for all $\boldsymbol{u}_{1},\!\boldsymbol{u}_{2}\in \mathbb{Z}_{2}^{k}$ with $f(\boldsymbol{u}_{1})\neq  f(\boldsymbol{u}_{2})$, it holds that $$d_{p}(Enc(\boldsymbol{u}_{1}),Enc(\boldsymbol{u}_{2}))\geq 2t+1.$$
\end{Definition}
\begin{Remark}{\rm
By this definition, given any pair vector $\mathop{\boldsymbol{y}}\limits^{\leftrightarrow}$, which is obtained by $t$-pair errors from $Enc(\boldsymbol{u})$, the receiver can uniquely recover $f(\boldsymbol{u})$, if it has knowledge about the function $f$ and $\mathop{\boldsymbol{y}}\limits^{\leftrightarrow}$.}
\end{Remark}
Analogous to \cite{LBWY},   we give the definition of the optimal redundancy of an FCSPC designed for the function $f$, the main quality of interest in this paper.
\begin{Definition}
  The optimal redundancy $r_{p}^{f}(k,t)$ is defined as the smallest $r$ such that there exists an FCSPC with encoding function $Enc\,:\,\mathbb{Z}_{2}^{k}\rightarrow\mathbb{Z}_{2}^{k+r}$ for $f$.
\end{Definition}
In order to find the optimal redundancy, we need to  connect FCSPCs with irregular-pair-distance codes.
For any integer $M$, we write $[M]^{+}\triangleq max\{M,0\}$ and recall that $[M]=\{1,\ldots,M\}$.
Then we introduced two pair-distance matrices of a function $f$.
\begin{Definition}
Let $\boldsymbol{u}_{1},\ldots,\boldsymbol{u}_{M}\in \mathbb{Z}_{2}^{k}.$
We define two kinds of pair-distance matrices $\boldsymbol{D}_{f}^{(1)}(t,\boldsymbol{u}_{1},\ldots,\boldsymbol{u}_{M})$ and $\boldsymbol{D}_{f}^{(2)}(t,\boldsymbol{u}_{1},\ldots,\boldsymbol{u}_{M})$ which are $M\times M$ matrices with entries
		$$\left[\boldsymbol{D}_f^{(1)}\left(t, \boldsymbol{u}_1, \ldots, \boldsymbol{u}_M\right)\right]_{i j}= \begin{cases}{\left[2t-d_{p}\left(\boldsymbol{u}_i, \boldsymbol{u}_j\right)\right]^{+},} & \text {if } f\left(\boldsymbol{u}_i\right) \neq f\left(\boldsymbol{u}_j\right), \\ 0, & \text { otherwise, }\end{cases}$$
and
		$$\ \left[\boldsymbol{D}_f^{(2)}\left(t, \boldsymbol{u}_1, \ldots, \boldsymbol{u}_M\right)\right]_{i j}= \begin{cases}{\left[2t+2-d_{p}\left(\boldsymbol{u}_i, \boldsymbol{u}_j\right)\right]^{+},} & \text {if } f\left(\boldsymbol{u}_i\right) \neq f\left(\boldsymbol{u}_j\right), \\ 0, & \text { otherwise. }\end{cases}$$
\end{Definition}
Let $\mathcal{P}=\{\boldsymbol{p}_{1},\boldsymbol{p}_{2},\ldots,\boldsymbol{p}_{M}\}\subseteq\mathbb{Z}_{2}^{r}$ be a code of length $r$ and cardinality $M$. Irregular-pair-distance codes are formally defined as follows.
\begin{Definition}
  Let $\boldsymbol{D}\in \mathbb{N}_{0}^{M\times M}$. If there exists an ordering of the codewords of $\mathcal{P}=\{\boldsymbol{p}_{1},\boldsymbol{p}_{2},\ldots,\boldsymbol{p}_{M}\}$ such that $d_{p}(\boldsymbol{p}_{i},\boldsymbol{p}_{j})\geq\left[\boldsymbol{D}\right]_{ij}$ for all $i,j\in \left[M\right]$,
  then $\mathcal{P}$ is called a $\boldsymbol{D}$-irregular-pair-distance code ($\boldsymbol{D}_{p}$-code for short).
  Moreover, let $N_{p}(\boldsymbol{D})$ be
   the   smallest integer $r$ such that  there exists a $\boldsymbol{D}_{p}$-code of length $r$.
 In the special case $[\boldsymbol{D}]_{ij}=D$ for all $i\neq j$,
  we simply  write $N_{p}(\boldsymbol{D})=N_{p}(M,D)$.
\end{Definition}
From the very definition, a $\boldsymbol{D}_{p}$-code requires individual pair-distance constraint between each pair of codewords.
With these definitions and facts, we can give the connection between FCSPCs and irregular-pair-distance codes.
\begin{Theorem}\label{theorem 3.1}
  For any function $f\,:\,\mathbb{Z}_{2}^{k}\rightarrow Im(f)$ and  $\{\boldsymbol{u}_{1},\ldots,\boldsymbol{u}_{2^{k}}\}=\mathbb{Z}_{2}^{k}$, we have
   	$$N_{p}\big(\boldsymbol{D}_{f}^{(1)}(t,\boldsymbol{u}_{1},\ldots,\boldsymbol{u}_{2^{k}})\big)\leq r_{p}^{f}(k,t)\leq N_{p}\big(\boldsymbol{D}_{f}^{(2)}(t,\boldsymbol{u}_{1},\ldots,\boldsymbol{u}_{2^{k}})\big).$$
\end{Theorem}
The following lemma plays an important role in the proving of Theorem \ref{theorem 3.1}.
\begin{lem}
  Let $\boldsymbol{u}=\left(\boldsymbol{u}^{(1)},\boldsymbol{u}^{(2)}\right)\in \mathbb{Z}_{2}^{m+r}$ and  $\boldsymbol{v}=\left(\boldsymbol{v}^{(1)},\boldsymbol{v}^{(2)}\right)\in \mathbb{Z}_{2}^{m+r}$, where $\boldsymbol{u}^{(1)}=\left(u_{1},\ldots,u_{m}\right)$, $\boldsymbol{u}^{(2)}=\left(u_{m+1},\ldots,u_{m+r}\right)$, $\boldsymbol{v}^{(1)}=\left(v_{1},\ldots,v_{m}\right)$ and  $\boldsymbol{v}^{(2)}=\left(v_{m+1},\ldots,v_{m+r}\right)$.
  Then
  $$d_{p}(\boldsymbol{u}^{(1)},\boldsymbol{v}^{(1)})+d_{p}(\boldsymbol{u}^{(2)},\boldsymbol{v}^{(2)})-1\leq d_{p}(\boldsymbol{u},\boldsymbol{v})\leq d_{p}(\boldsymbol{u}^{(1)},\boldsymbol{v}^{(1)})+d_{p}(\boldsymbol{u}^{(2)},\boldsymbol{v}^{(2)})+1.$$
\end{lem}
\begin{proof}
  According to the definition of pair-distance,
  $d_{p}(\boldsymbol{u},\boldsymbol{v})=|\{i\in[m+r]\,|\,(u_{i},u_{i+1})\neq(v_{i},v_{i+1})\}|$, where
  the subscripts are taken modulo $m+r$. Similarly, we have
  $$d_{p}(\boldsymbol{u}^{(1)},\boldsymbol{v}^{(1)})=|\{i\in[m]\,:\,(u_{i},u_{i+1})\neq(v_{i},v_{i+1})\}|,$$
where the subscripts are taken  modulo $m$;
  $$d_{p}(\boldsymbol{u}^{(2)},\boldsymbol{v}^{(2)})=|\{j\in[r]\,:\,(u_{m+j},u_{m+j+1})\neq(v_{m+j},v_{m+j+1})\}|,$$
  where the subscripts   $js$ are taken modulo $r$.
  Let $\widetilde{\boldsymbol{u}}=((u_{m},u_{m+1}),(u_{m+r},u_{1}))$, $\widehat{\boldsymbol{u}}=((u_{m},u_{1}),(u_{m+r},u_{m+1}))$, $\widetilde{\boldsymbol{v}}=((v_{m},v_{m+1}),(v_{m+r},v_{1}))$ and  $\widehat{\boldsymbol{v}}=((v_{m},v_{1}),(v_{m+r},v_{m+1}))$. It follows that $$d_{p}(\boldsymbol{u},\boldsymbol{v})-d_{p}(\boldsymbol{u}^{(1)},\boldsymbol{v}^{(1)})-d_{p}(\boldsymbol{u}^{(2)},\boldsymbol{v}^{(2)})
  =d_{H}(\widetilde{\boldsymbol{u}},\widetilde{\boldsymbol{v}})-d_{H}(\widehat{\boldsymbol{u}},\widehat{\boldsymbol{v}}).$$
  Since $0\leq d_{H}(\widetilde{\boldsymbol{u}},\widetilde{\boldsymbol{v}})\leq 2$ and  $0\leq d_{H}(\widehat{\boldsymbol{u}},\widehat{\boldsymbol{v}})\leq 2$, one has
  $-2\leq d_{H}(\widetilde{\boldsymbol{u}},\widetilde{\boldsymbol{v}})-d_{H}(\widehat{\boldsymbol{u}},\widehat{\boldsymbol{v}})\leq 2.$
If $d_{H}(\widetilde{\boldsymbol{u}},\widetilde{\boldsymbol{v}})-d_{H}(\widehat{\boldsymbol{u}},\widehat{\boldsymbol{v}})=2$,
then
$d_{H}(\widetilde{\boldsymbol{u}},\widetilde{\boldsymbol{v}})=2\ \text{and}\ d_{H}(\widehat{\boldsymbol{u}},\widehat{\boldsymbol{v}})=0.$
  Since $d_{H}(\widehat{\boldsymbol{u}},\widehat{\boldsymbol{v}})=0$, $$u_{1}=v_{1},\,u_{m}=v_{m},\,u_{m+1}=v_{m+1},\,u_{m+r}=v_{m+r}.$$
  Thus, $d_{H}(\widetilde{\boldsymbol{u}},\widetilde{\boldsymbol{v}})=0$,  contradicting to the fact $d_{H}(\widetilde{\boldsymbol{u}},\widetilde{\boldsymbol{v}})=2$.
Similarly, $d_{H}(\widetilde{\boldsymbol{u}},\widetilde{\boldsymbol{v}})-d_{H}(\widehat{\boldsymbol{u}},\widehat{\boldsymbol{v}})\neq -2.$
This leads to
$-1\leq d_{H}(\widetilde{\boldsymbol{u}},\widetilde{\boldsymbol{v}})-d_{H}(\widehat{\boldsymbol{u}},\widehat{\boldsymbol{v}})\leq 1.$
Therefore, we conclude that
$$d_{p}(\boldsymbol{u}^{(1)},\boldsymbol{v}^{(1)})+d_{p}(\boldsymbol{u}^{(2)},\boldsymbol{v}^{(2)})-1\leq d_{p}(\boldsymbol{u},\boldsymbol{v})\leq d_{p}(\boldsymbol{u}^{(1)},\boldsymbol{v}^{(1)})+d_{p}(\boldsymbol{u}^{(2)},\boldsymbol{v}^{(2)})+1.$$
\end{proof}
We are now in a position to prove Theorem \ref{theorem 3.1}.

\noindent
{\bf Proof of Theorem \ref{theorem 3.1}}~
We first claim that $r_{p}^{f}(k,t)\geq N_{p}(\boldsymbol{D}_{f}^{(1)}(t,\boldsymbol{u}_{1},\ldots,\boldsymbol{u}_{2^{k}})).$
Suppose otherwise that  $r_{p}^{f}(k,t)<N_{p}(\boldsymbol{D}_{f}^{(1)}(t,\boldsymbol{u}_{1},\ldots,\boldsymbol{u}_{2^{k}})).$		
Let $Enc\,:\,\mathbb{Z}_{2}^{k}\rightarrow\mathbb{Z}_{2}^{k+r}\,,\,\boldsymbol{u}_{i}\mapsto\left(\boldsymbol{u}_{i},\boldsymbol{p}_{i}\right)$ with $r=r_{p}^{f}(k,t)$ define an FCSPC for the function $f$. Then, there must exist $i,j\in [2^{k}]$ with $f(\boldsymbol{u}_{i})\neq f(\boldsymbol{u}_{j})$ such that $d_{p}(\boldsymbol{p}_{i},\boldsymbol{p}_{j})<2t-d_{p}(\boldsymbol{u}_{i},\boldsymbol{u}_{j}).$
This forces
$$d_{p}(Enc(\boldsymbol{u}_{i}),Enc(\boldsymbol{u}_{j}))\leq d_{p}(\boldsymbol{u}_{i},\boldsymbol{u}_{j})+d_{p}(\boldsymbol{p}_{i},\boldsymbol{p}_{j})+1<2t+1,$$
which contradicts to the definition of FCSPC. We thus have $r_{p}^{f}(k,t)\geq N_{p}(\boldsymbol{D}_{f}^{(1)}(t,\boldsymbol{u}_{1},\ldots,\boldsymbol{u}_{2^{k}}))$, as claimed.

Then, we aim to show that $r_{p}^{f}(k,t)\leq N_{p}(\boldsymbol{D}_{f}^{(2)}(t,\boldsymbol{u}_{1},\ldots,\boldsymbol{u}_{2^{k}}))$.
For this purpose,
let $\mathcal{P}=\{\boldsymbol{p}_{1},\ldots,\boldsymbol{p}_{2^{k}}\}$ be a $\boldsymbol{D}_{f}^{(2)}(t,\boldsymbol{u}_{1},\ldots,\boldsymbol{u}_{2^{k}})_{p}$-code of length $N_{p}(\boldsymbol{D}_{f}^{(2)}(t,\boldsymbol{u}_{1},\ldots,\boldsymbol{u}_{2^{k}}))$ and define the encoding function
	$Enc\,:\,\mathbb{Z}_{2}^{k}\rightarrow\mathbb{Z}_{2}^{k+r}\,,\,\boldsymbol{u}_{i}\mapsto(\boldsymbol{u}_{i},\boldsymbol{p}_{i})$.	
For any $i,j\in [2^{k}]$ with $f(\boldsymbol{u}_{i})\neq f(\boldsymbol{u}_{j})$, one has
$$d_{p}(Enc(\boldsymbol{u}_{i}),Enc(\boldsymbol{u}_{j}))\geq d_{p}(\boldsymbol{u}_{i},\boldsymbol{u}_{j})+d_{p}(\boldsymbol{p}_{i},\boldsymbol{p}_{j})-1\geq 2t+2-1=2t+1.$$
Thus, the encoding function $Enc$ defines an FCSPC for the function $f$. Hence, $r_{p}^{f}(k,t)\leq N_{p}(\boldsymbol{D}_{f}^{(2)}(t,\boldsymbol{u}_{1},\ldots,\boldsymbol{u}_{2^{k}})).$\qed\\
Theorem \ref{theorem 3.1} provides  upper and lower bounds on the optimal redundancy of FCSPCs designed for generic functions
and transforms the computation of $r_{p}^{f}(k,t)$ to the computation of $N_{p}(\boldsymbol{D}_{f}^{(1)}(t,\boldsymbol{u}_{1},\ldots,\boldsymbol{u}_{2^{k}}))$ and $N_{p}(\boldsymbol{D}_{f}^{(2)}(t,\boldsymbol{u}_{1},\ldots,\boldsymbol{u}_{2^{k}}))$.
If the function $f$ is a constant function,
then $$N_{p}(\boldsymbol{D}_{f}^{(1)}(t,\boldsymbol{u}_{1},\ldots,\boldsymbol{u}_{2^{k}}))=\\r_{p}^{f}(k,t)=
N_{p}(\boldsymbol{D}_{f}^{(2)}(t,\boldsymbol{u}_{1},\ldots,\boldsymbol{u}_{2^{k}}))=0.$$
\subsection{Simplified Bounds of the Optimal Redundancy}
Theorem \ref{theorem 3.1} provides upper and lower bounds for the optimal redundancy, depending on the pair-distance matrices $\boldsymbol{D}_{f}^{(1)}(t,\boldsymbol{u}_{1},\ldots,\boldsymbol{u}_{2^{k}})$ and $\boldsymbol{D}_{f}^{(2)}(t,\boldsymbol{u}_{1},\ldots,\boldsymbol{u}_{2^{k}})$; however, it is not an easy task to get the true values of $N_{p}(\boldsymbol{D}_{f}^{(1)}(t,\boldsymbol{u}_{1},\ldots,\boldsymbol{u}_{2^{k}}))$ and $N_{p}(\boldsymbol{D}_{f}^{(2)}(t,\boldsymbol{u}_{1},\ldots,\boldsymbol{u}_{2^{k}}))$  in general. For ease of evaluation, we derive simplified and possibly sub-optimal bounds by considering   subsets of $\mathbb{Z}_{2}^{k}$.

As a first corollary of Theorem \ref{theorem 3.1}, we find a simplified lower bound on $r_{p}^{f}(k,t)$ which acts on arbitrary subset of $\mathbb{Z}_{2}^{k}$.
\begin{Corollary}\label{corollary 3.1}
  Let $\{\boldsymbol{u}_{1},\ldots,\boldsymbol{u}_{M}\}$ be a subset of $\mathbb{Z}_{2}^{k}$.
  Then, the redundancy of an FCSPC for $f$ satisfies
 $$r_{p}^{f}(k,t)\geq N_{p}(\boldsymbol{D}_{f}^{(1)}\big(t,\boldsymbol{u}_{1},\ldots,\boldsymbol{u}_{M})\big).$$
Particularly, for any function $f$ with $|Im(f)|\geq 2$, $r_{p}^{f}(k,t)\geq 2t-2$.
\end{Corollary}
\begin{proof}
  Let $\{\boldsymbol{u}_{1},\ldots,\boldsymbol{u}_{M},\boldsymbol{u}_{M+1},\ldots,\boldsymbol{u}_{2^{k}}\}=\mathbb{Z}_{2}^{k}$
  and $\{\boldsymbol{p}_{1},\ldots,\boldsymbol{p}_{M},\boldsymbol{p}_{M+1},\ldots,\boldsymbol{p}_{2^{k}}\}$ be a $\boldsymbol{D}_{f}^{(1)}(t,\boldsymbol{u}_{1},\ldots,\boldsymbol{u}_{2^{k}})$-code of length $N_{p}(\boldsymbol{D}_{f}^{(1)}(t,\boldsymbol{u}_{1},\ldots,\boldsymbol{u}_{2^{k}}))$,
  then $\{\boldsymbol{p}_{1},\ldots,\boldsymbol{p}_{M}\}$ is a $\boldsymbol{D}_{f}^{(1)}(t,\boldsymbol{u}_{1},\ldots,\boldsymbol{u}_{M})$-code.
  Thus, $$N_{p}(\boldsymbol{D}_{f}^{(1)}(t,\boldsymbol{u}_{1},\ldots,\boldsymbol{u}_{2^{k}}))\geq N_{p}(\boldsymbol{D}_{f}^{(1)}(t,\boldsymbol{u}_{1},\ldots,\boldsymbol{u}_{M})).$$
  By Theorem \ref{theorem 3.1}, $r_{p}^{f}(k,t)\geq N_{p}(\boldsymbol{D}_{f}^{(1)}(t,\boldsymbol{u}_{1},\ldots,\boldsymbol{u}_{2^{k}})).$
Hence, $r_{p}^{f}(k,t)\geq N_{p}(\boldsymbol{D}_{f}^{(1)}(t,\boldsymbol{u}_{1},\ldots,\boldsymbol{u}_{M})).$
Since $|Im(f)|\geq2$, there exists $\boldsymbol{u},\,\boldsymbol{u}^{\prime}\in\mathbb{Z}_{2}^{k}$ such that
$f(\boldsymbol{u})\neq f(\boldsymbol{u}^{\prime})$ and $d_{p}(\boldsymbol{u},\boldsymbol{u}^{\prime})=2$.
By the conclusion above, $$r_{p}^{f}(k,t)\geq N_{p}(\boldsymbol{D}_{f}^{(1)}(t,\boldsymbol{u},\boldsymbol{u}^{\prime}))=N_{p}(2,2t-2).$$
$N_{p}(2,2t-2)=2t-2$, which is attained by the repetition code $\{(0,\ldots,0),(1,\ldots,1)\}$ of length $2t-2$.
\end{proof}
We now provide a simplified upper bound on $r_{p}^{f}(k,t)$ which is obtained by focusing on a representative subset of information vectors.
These  information vectors have different function values and are close in pair-distance. To make things clear, we give  the following definitions.
\begin{Definition}
  The pair-distance between two function values $f_{1},f_{2}\in Im(f)$ is defined as the smallest pair-distance between two information vectors that evaluate to $f_{1}$ and $f_{2}$,  i.e.,
  $$d_{p}^{f}\left(f_1, f_2\right) \triangleq \min _{\boldsymbol{u}_1, \boldsymbol{u}_2 \in \mathbb{Z}_2^k} d_{p}\left(\boldsymbol{u}_1, \boldsymbol{u}_2\right) \text { s.t. } f\left(\boldsymbol{u}_1\right)=f_1~\hbox{and}~f\left(\boldsymbol{u}_2\right)=f_2.$$
\end{Definition}
\begin{Definition}
  The function pair-distance matrices of a function are denoted by the $E\times E$ matrices $\boldsymbol{D}_{f}^{(1)}(t,f_{1},\ldots,f_{E})$ and $\boldsymbol{D}_{f}^{(2)}(t,f_{1},\ldots,f_{E})$ with entries $$\left[\boldsymbol{D}_{f}^{(1)}\left(t, f_1, \ldots, f_E\right)\right]_{i j}= \begin{cases}{\left[2 t-d_{p}^{f}\left(f_i, f_j\right)\right]^{+},} & \text {if } i\neq j, \\ 0, & \text { otherwise, }\end{cases}$$
   and
   $$\left[\boldsymbol{D}_{f}^{(2)}\left(t, f_1, \ldots, f_E\right)\right]_{i j}= \begin{cases}{\left[2 t+2-d_{p}^{f}\left(f_i, f_j\right)\right]^{+},} & \text {if } i\neq j, \\ 0, & \text { otherwise,}\end{cases}$$
   where $E=|Im(f)|$.
\end{Definition}
With these definitions, we can give an upper bound for the optimal redundancy.
\begin{Theorem}\label{theorem 3.2}
  For arbitrary function $f\,:\,\mathbb{Z}_{2}^{k}\rightarrow Im(f)$, we have
  $r_{p}^{f}(k,t)\leq N_{p}(\boldsymbol{D}_{f}^{(2)}\left(t, f_1, \ldots, f_E\right)).$
\end{Theorem}
\begin{proof}
  We describe the process of  how to construct an FCSPC for the function $f$.
Denote by $\boldsymbol{p}_{i}$ the redundancy vector assigned to all
$\boldsymbol{u}\in\mathbb{Z}_{2}^{k}$ with $f(\boldsymbol{u})=f_{i}.$
Therefore, two information vectors with the same function value have the same redundancy vector.
Then, we choose $\boldsymbol{p}_{1},\ldots,\boldsymbol{p}_{E}$ such that $$d_{p}(\boldsymbol{p}_{i},\boldsymbol{p}_{j})\geq 2t+2-d_{p}^{f}(f_{i},f_{j}),\,i,j\in[E].$$  It follows that for any $\boldsymbol{u}_{i},\,\boldsymbol{u}_{j}\in\mathbb{Z}_{2}^{k}$ with $f(\boldsymbol{u}_{i})=f_{i}$, $f(\boldsymbol{u}_{j})=f_{j}$, $f_{i}\neq f_{j}$, we have $$d_{p}(Enc(\boldsymbol{u}_{i}),Enc(\boldsymbol{u}_{j}))\geq d_{p}(\boldsymbol{u}_{i},\boldsymbol{u}_{j})+d_{p}(\boldsymbol{p}_{i},\boldsymbol{p}_{j})-1\geq 2t+2-1=2t+1.$$According to the definition of FCSPC, we have constructed an FCSPC. We also guarantee the existence of such parity vectors $\boldsymbol{p}_{1},\ldots,\boldsymbol{p}_{E}$ if they have length $N_{p}(\boldsymbol{D}_{f}^{(2)}\left(t, f_1, \ldots, f_E\right)).$
  Hence, $r_{p}^{f}(k,t)\leq N_{p}(\boldsymbol{D}_{f}^{(2)}\left(t, f_1, \ldots, f_E\right)).$
\end{proof}
Since the sizes of $\boldsymbol{D}_{f}^{(1)}\big(t,\boldsymbol{u}_{1},\ldots,\boldsymbol{u}_{M})$ and $\boldsymbol{D}_{f}^{(2)}\left(t, f_1, \ldots, f_E\right)$ are smaller than those of  $\boldsymbol{D}_{f}^{(1)}(t,\boldsymbol{u}_{1},\ldots,\boldsymbol{u}_{2^{k}})$ and $\boldsymbol{D}_{f}^{(2)}(t,\boldsymbol{u}_{1},\ldots,\boldsymbol{u}_{2^{k}})$,
the values of $N_{p}(\boldsymbol{D}_{f}^{(1)}\big(t,\boldsymbol{u}_{1},\ldots,\boldsymbol{u}_{M}))$ and
$N_{p}(\boldsymbol{D}_{f}^{(2)}\left(t, f_1, \ldots, f_E\right))$ are easier to evaluate in general.
\subsection{Bounds on $N_{p}(\boldsymbol{D})$}
In order to obtain the optimal redundancy by applying the results above, we must get the value of $N_{p}(\boldsymbol{D})$.
Though it is difficult to get the exact value of $N_{p}(\boldsymbol{D})$, we can provide some upper and lower bounds of $N_{p}(\boldsymbol{D})$. For convenience, we firstly give a definition.
\begin{Definition}
  We call a symmetric matrix $\boldsymbol{D}\in\mathbb{N}_{0}^{M\times M}$ distance matrix if $[\boldsymbol{D}]_{ij}\geq 0$ and $[\boldsymbol{D}]_{ii}=0$.
\end{Definition}
Then, we present a lower bound on $N_{p}(\boldsymbol{D})$.
\begin{lem}\label{lemma 3.2}
  For any distance matrix $\boldsymbol{D}\in \mathbb{N}_{0}^{M\times M}$, we have
  $$N_p(\boldsymbol{D})\geq \begin{cases}{\frac{8}{3(M^{2}-1)}\sum\limits_{i,j,i<j}[\boldsymbol{D}]_{ij},} & \text{M is odd},\\
			{\frac{8}{3M^{2}}\sum\limits_{i,j,i<j}[\boldsymbol{D}]_{ij},} & M\equiv 0(\bmod 4),\\
			{\frac{8}{3M^{2}-4}\sum\limits_{i,j,i<j}[\boldsymbol{D}]_{ij},} & M\equiv 2(\bmod 4).
		 \end{cases}$$
\end{lem}
\begin{proof}
  Let $\mathcal{P}=\{\boldsymbol{p}_{1},\ldots,\boldsymbol{p}_{M}\}$ be a $\boldsymbol{D}_{p}$-code of length $r$ and $\boldsymbol{P}$ be the $M\times r$ matrix whose rows are the symbol-pair read vectors of codewords of $\mathcal{P}$.
Let $S=\sum\limits_{i,j}d_{p}(\boldsymbol{p}_{i},\boldsymbol{p}_{j})$ and $n_{k(\alpha,\beta)}$ be the number of times $(\alpha,\beta)\in\mathbb{Z}_{2}\times\mathbb{Z}_{2}$ occurs in the $k$th column of $\boldsymbol{P}$.
Then, \begin{equation*}
			\setlength{\abovedisplayskip}{3pt}
			\setlength{\belowdisplayskip}{4pt}
			\begin{aligned}
				S  = & \sum_{k=1}^{r}[n_{k(0,0)}(M-n_{k(0,0)})+n_{k(0,1)}(M-n_{k(0,1)})+n_{k(1,0)}(M-n_{k(1,0)})+n_{k(1,1)}(M-n_{k(1,1)})]\\
				= & \sum_{k=1}^{r}[M(n_{k(0,0)}+n_{k(0,1)}+n_{k(1,0)}+n_{k(1,1)})-n_{k(0,0)}^{2}--n_{k(0,1)}^{2}-n_{k(1,0)}^{2}-n_{k(1,1)}^{2}]\\
				 = &\sum_{k=1}^{r}[M^{2}-n_{k(0,0)}^{2}--n_{k(0,1)}^{2}-n_{k(1,0)}^{2}-n_{k(1,1)}^{2}]\\
				 = & \sum_{k=1}^{r}[M^{2}-\sum_{(\alpha,\beta)\in\mathbb{Z}_{2}\times\mathbb{Z}_{2}}n_{k(\alpha,\beta)}^{2}].
			\end{aligned}
\end{equation*}
For each $1\leq k\leq r$, one has
		\begin{equation*}
			\setlength{\abovedisplayskip}{4pt}
			\setlength{\belowdisplayskip}{3pt}
			\begin{aligned}
				M^{2}= & [\sum_{(\alpha,\beta)\in\mathbb{Z}_{2}\times\mathbb{Z}_{2}}n_{k(\alpha,\beta)}]^{2}\\ = & \sum_{(\alpha,\beta)\in\mathbb{Z}_{2}\times\mathbb{Z}_{2}}n_{k(\alpha,\beta)}^{2}+2n_{k(0,0)}n_{k(0,1)}+2n_{k(0,0)}n_{k(1,0)}+ 2n_{k(0,0)}n_{k(1,1)}\\ & +2n_{k(0,1)}n_{k(1,0)}+2n_{k(0,1)}n_{k(1,1)}+2n_{k(1,0)}n_{k(1,1)}.
			\end{aligned}
		\end{equation*}
Then,\begin{equation*}
		\setlength{\abovedisplayskip}{3pt}
		\setlength{\belowdisplayskip}{3pt}
		\begin{aligned}
			S= & \sum_{k=1}^{r}[2n_{k(0,0)}n_{k(0,1)}+2n_{k(0,0)}n_{k(1,0)}+2n_{k(0,0)}n_{k(1,1)}+\\ & 2n_{k(0,1)}n_{k(1,0)}+2n_{k(0,1)}n_{k(1,1)}+2n_{k(1,0)}n_{k(1,1)}].
		\end{aligned}
	\end{equation*}
Hence,\begin{equation*}
			\setlength{\abovedisplayskip}{3pt}
			\setlength{\belowdisplayskip}{5pt}
			\begin{aligned}
				\sum_{i,j,i<j}d_{p}(\boldsymbol{p}_{i},\boldsymbol{p}_{j}) =\frac{1}{2}S
				= & \sum_{k=1}^{r}[n_{k(0,0)}n_{k(0,1)}+n_{k(0,0)}n_{k(1,0)}+n_{k(0,0)}n_{k(1,1)}\\
				& +n_{k(0,1)}n_{k(1,0)}+n_{k(0,1)}n_{k(1,1)}+n_{k(1,0)}n_{k(1,1)}].
			\end{aligned}
		\end{equation*}

If $M\equiv 0 \pmod 4$, the right hand side of the equation above is maximized when $\{n_{k(0,0)},n_{k(0,1)},n_{k(1,0)},n_{k(1,1)}\}=\{\frac{M}{4},\frac{M}{4},\frac{M}{4},\frac{M}{4}\}.$
	Thus, \begin{equation}\nonumber
		\setlength{\abovedisplayskip}{4pt}
		\setlength{\belowdisplayskip}{4pt}
		\sum_{i,j,i<j}[\boldsymbol{D}]_{ij}\leq\sum_{i,j,i<j}d_{p}(\boldsymbol{p}_{i},\boldsymbol{p}_{j})\leq\frac{3}{8}M^{2}r.
	\end{equation}
We have \begin{equation}\nonumber
		\setlength{\abovedisplayskip}{4pt}
		\setlength{\belowdisplayskip}{4pt}
		N_{p}(\boldsymbol{D})\geq\frac{8}{3M^{2}}\sum_{i,j,i<j}[\boldsymbol{D}]_{ij}.
	\end{equation}
If $M\equiv 1 \pmod 4$, the right hand side of the equation above is maximized when $\{n_{k(0,0)},n_{k(0,1)},n_{k(1,0)},n_{k(1,1)}\}=\{\frac{M-1}{4},\frac{M-1}{4},\frac{M-1}{4},\frac{M+3}{4}\}.$ Similarly, we have \begin{equation}\nonumber
			\setlength{\abovedisplayskip}{4pt}
			\setlength{\belowdisplayskip}{4pt}
			N_{p}(\boldsymbol{D})\geq\frac{8}{3(M^{2}-1)}\sum_{i,j,i<j}[\boldsymbol{D}]_{ij}.
		\end{equation}
	If $M\equiv 2 \pmod 4$, the right hand side of the equation above is maximized when $\{n_{k(0,0)},n_{k(0,1)},n_{k(1,0)},n_{k(1,1)}\}=\{\frac{M-2}{4},\frac{M-2}{4},\frac{M+2}{4},\frac{M+2}{4}\}.$ This leads to
 \begin{equation}\nonumber
		\setlength{\abovedisplayskip}{4pt}
		\setlength{\belowdisplayskip}{4pt}
		N_{p}(\boldsymbol{D})\geq\frac{8}{3M^{2}-4}\sum_{i,j,i<j}[\boldsymbol{D}]_{ij}.
	\end{equation}
	If $M\equiv 3 \pmod 4$, the right hand side of the equation above is maximized when $\{n_{k(0,0)},n_{k(0,1)},n_{k(1,0)},n_{k(1,1)}\}=\{\frac{M+1}{4},\frac{M+1}{4},\frac{M+1}{4},\frac{M-3}{4}\}.$  We conclude that \begin{equation}\nonumber
		\setlength{\abovedisplayskip}{4pt}
		\setlength{\belowdisplayskip}{2pt}
		N_{p}(\boldsymbol{D})\geq\frac{8}{3(M^{2}-1)}\sum_{i,j,i<j}[\boldsymbol{D}]_{ij}.
	\end{equation}
\end{proof}
\begin{Remark}{\rm
  In fact, Lemma \ref{lemma 3.2} is a generalization of the Plotkin's bound. Indeed, take $[\boldsymbol{D}]_{ij}=D$ for all $i\neq j$, Lemma \ref{lemma 3.2} implies $N_{p}(M,D)\geq\frac{4(M-1)}{3M}D$, a variant of Plotkin's bound for symbol-pair codes.}
\end{Remark}
Apart from the lower bound, we also find an upper bound on $N_{p}(\boldsymbol{D})$.
\begin{lem}\label{lemma 3.3}
  For any distance matrix $\boldsymbol{D}\in\mathbb{N}_{0}^{M\times M}$ and any permutation $\pi\,:\,[M]\rightarrow[M]$, $$N_{p}(\boldsymbol{D})\leq\min_{r\in\mathbb{N}}\{r\,:\,2^{r}>\max_{j\in[M]}\sum_{i=1}^{j-1}B(r,[\boldsymbol{D}]_{\pi(i)\pi(j)-1})\},$$
	where $B(r,d)$ is the number of words with pair-distance $d$ or less from $\boldsymbol{x}\in\mathbb{Z}_{2}^{r}$.
\end{lem}
\begin{proof}
  We construct a $\boldsymbol{D}_{p}$-code of length $r$ by iteratively selecting valid codewords.
 Assume first for simplicity that $\pi(i)=i$. Start by choosing an arbitrary codeword $\boldsymbol{p}_{1}\in\mathbb{Z}_{2}^{r}$.
Then, choose a valid codeword $\boldsymbol{p}_{2}$ as follows: Since the pair-distance of $\boldsymbol{p}_{1}$ and $\boldsymbol{p}_{2}$ needs to be at least $[\boldsymbol{D}]_{12}$, we choose an arbitrary $\boldsymbol{p}_{2}$ such that $d_{p}(\boldsymbol{p}_{1},\boldsymbol{p}_{2})\geq [\boldsymbol{D}]_{12}$. Such a codeword $\boldsymbol{p}_{2}$ exists, if the length satisfies $2^{r}>B(r,[\boldsymbol{D}]_{12}-1)$.

  Next, we choose the third codeword $\boldsymbol{p}_{3}$. Similarly as before, we need to have $d_{p}(\boldsymbol{p}_{1},\boldsymbol{p}_{3})\geq [\boldsymbol{D}]_{13}$ and $d_{p}(\boldsymbol{p}_{2},\boldsymbol{p}_{3})\geq [\boldsymbol{D}]_{23}$. If $2^{r}>B(r,[\boldsymbol{D}]_{13}-1)+B(r,[\boldsymbol{D}]_{23}-1)$, we can guarantee the existence of such a codeword $\boldsymbol{p}_{3}$.

  The lemma then follows by iteratively selecting the remaining codewords $\boldsymbol{p}_{j}$ such that $d_{p}(\boldsymbol{p}_{i},\boldsymbol{p}_{j})\geq [\boldsymbol{D}]_{ij}$ for all $i<j$. Under the condition of the lemma, we can guarantee the existence of all codewords. Since the codewords can be chosen in an arbitrary order, the lemma holds for any order $\pi$ in which the codewords are selected.
\end{proof}
\begin{Remark}{\rm
Lemma \ref{lemma 3.3} is a generalization of well-known Gilbert-Varshamov bound. Taking $[\boldsymbol{D}]_{ij}=D$, we get
 $N_{p}(M,D)\leq \min\limits_{r\in\mathbb{N}}\{r\,|\,2^{r}>(M-1)B(r,D-1)\}$ , a variant of the Gilbert-Varshamov bound.}
\end{Remark}
The next lemma connects $N_{p}(M,D)$ with $N(M,D)$.
\begin{lem}\label{hammingsymbol}
  Let $M,D$ be positive integers with $M> 2$, $D\geq 2$, we have  $N_{p}(M,D)\leq N(M,D-1).$
\end{lem}
\begin{proof}
  	Let $r=N(M,D-1)$, then there exists a code $\mathcal{P}=\{\boldsymbol{p}_{1},\ldots,\boldsymbol{p}_{M}\}$ of length $r$ such that $d_{H}(\boldsymbol{p}_{i},\boldsymbol{p}_{j})\geq D-1$ for any $i,j\in [M]$ with $i\neq j$.
By the Singleton bound of Hamming metric \cite{HP}, $2^{r-(D-1)+1}\geq M>2$,  giving $r>D-1.$
If $d_{H}(\boldsymbol{p}_{i},\boldsymbol{p}_{j})<r$, then $$d_{p}(\boldsymbol{p}_{i},\boldsymbol{p}_{j})\geq d_{H}(\boldsymbol{p}_{i},\boldsymbol{p}_{j})+1\geq D-1+1=D.$$
If $d_{H}(\boldsymbol{p}_{i},\boldsymbol{p}_{j})=r$, then $d_{p}(\boldsymbol{p}_{i},\boldsymbol{p}_{j})=d_{H}(\boldsymbol{p}_{i},\boldsymbol{p}_{j})=r\geq D.$
It follows that $N_{p}(M,D)\leq r=N(M,D-1).$
\end{proof}
\begin{Remark}{\rm
By Lemma \ref{lemma 2.3}, $d_{H}(\boldsymbol{x},\boldsymbol{y})+1\leq d_{p}(\boldsymbol{x},\boldsymbol{y})\leq 2d_{H}(\boldsymbol{x},\boldsymbol{y})$ except in the extreme cases. When $M=2$, we cannot avoid the extreme cases, so we cannot have $N_{p}(M,D)\leq N(M,D-1).$ In fact, when $D\geq 2$, $N_{p}(M,D)=D$ and $N(M,D-1)=D-1$.}
\end{Remark}
Combining Lemmas \ref{lemma 2.1}, \ref{lemma 2.2} and \ref{hammingsymbol}, we can derive two upper bounds of $N_{p}(M,D)$ easily.
\begin{lem}\label{lemma 3.5}
  Let $D\geq 2$ be a positive integer such that there exists a Hadamard matrix of order $2D-2$ and $M\leq 4D-4$. We then have $N_{p}(M,D)\leq 2D-2.$
\end{lem}
\begin{lem}\label{lemma 3.6}
  For any positive integers $M,D$ with $D\geq 11$ and $M\leq (D-1)^{2}$, $$N_{p}(M,D)\leq\frac{2D-4}{1-2\sqrt{ln(D-1)/(D-1)}}.$$
\end{lem}
\begin{Remark}{\rm
By\cite{H}, a positive integer $n$ can be the order of a Hadamard matrix only if $n=1$, $n=2$ or $n$ is a multiple of $4$. Therefore, the two upper bounds can be used in different cases.}
\end{Remark}
Next, using the relation between $N_{p}(M,D)$ and $N_{p}(\boldsymbol{D})$, we can get an upper bounds of $N_{p}(\boldsymbol{D})$,
as given below.
\begin{lem}
  Let $\boldsymbol{D}\in\mathbb{N}_{0}^{M\times M}$ and $D_{max}=\max\limits_{i,j}[\boldsymbol{D}]_{ij}$. If $D_{max}\leq D$, then $N_{p}(\boldsymbol{D})\leq N_{p}(M,D).$
\end{lem}
\begin{proof}
  Let $r=N_{p}(M,D)$ and let $\mathcal{C}=\{\boldsymbol{p}_{1},\ldots,\boldsymbol{p}_{M}\}$ be a code of length $r$   satisfying $d_{p}(\boldsymbol{p}_{i},\boldsymbol{p}_{j})\geq D$ when $i\neq j$. Since $D_{max}\leq D$, $$d_{p}(\boldsymbol{p}_{i},\boldsymbol{p}_{j})\geq D\geq D_{max}\geq [\boldsymbol{D}]_{ij}.$$ Hence, $\mathcal{C}$ is a $\boldsymbol{D}_{p}$-code. Therefore, $N_{p}(\boldsymbol{D})\leq N_{p}(M,D).$
\end{proof}
With Lemmas \ref{lemma 3.5} and \ref{lemma 3.6}, we get the following two corollaries easily.
\begin{Corollary}
  Let $\boldsymbol{D}\in\mathbb{N}_{0}^{M\times M}$ and $D_{max}=\max\limits_{i,j}[\boldsymbol{D}]_{ij}$. If $D_{max}\leq D$ where   $D\geq 2$
  is a positive integer such that there exists a Hadamard matrix of order $2D-2$ and $M\leq 4D-4$,
  then
  $N_{p}(\boldsymbol{D})\leq 2D-2.$
\end{Corollary}
\begin{Corollary}\label{corollary 3.3}
  Let $\boldsymbol{D}\in\mathbb{N}_{0}^{M\times M}$ and $D_{max}=\max\limits_{i,j}[\boldsymbol{D}]_{ij}$. If $D_{max}\leq D$ where   $D\geq 11$
  is a positive integer and $M\leq (D-1)^{2}$, then
  $$N_{p}(\boldsymbol{D})\leq\frac{2D-4}{1-2\sqrt{ln(D-1)/(D-1)}}.$$
\end{Corollary}
In the following sections, we will discuss some specific functions. We provide some upper and lower bounds of the optimal redundancy of FCSPCs designed for these functions and some explicit constructions of FCSPCs.

\section{Pair-Locally Binary Functions}
This section is centred around pair-locally binary functions, a new class of functions which is extended from locally binary functions. Firstly, we give the definitions of such functions.
\begin{Definition}
  The function ball of a function $f$ with pair-radius $\rho$ around $\boldsymbol{u}\in\mathbb{Z}_{2}^{k}$ is defined by $$B_{p}^{f}(\boldsymbol{u},\rho)=\{f(\boldsymbol{u}^{\prime})\,|\,\boldsymbol{u}^{\prime}\in \mathbb{Z}_{2}^{k}~\hbox{and}~ d_{p}(\boldsymbol{u},\boldsymbol{u}^{\prime})\leq\rho\}.$$
\end{Definition}
\begin{Definition}
  A function $f\,:\,\mathbb{Z}_{2}^{k}\rightarrow Im(f)$ is called a $\rho$-pair-locally binary function if for all $\boldsymbol{u}\in\mathbb{Z}_{2}^{k}$, $|B_{p}^{f}(\boldsymbol{u},\rho)|\leq 2.$
\end{Definition}
Then, we derive the optimal redundancy of FCSPCs designed for pair-locally binary functions and provide an explicit construction.
\begin{lem}
  	For any $2t$-pair-locally binary function, we have
  $2t-2\leq r_{p}^{f}(k,t)\leq 2t-1.$
\end{lem}
\begin{proof}
It follows from Corollary \ref{corollary 3.1} that  $r_{p}^{f}(k,t)\geq 2t-2.$
We show  the upper bound holds true by using the following explicit code construction.
For $\boldsymbol{u}\in\mathbb{Z}_{2}^{k}$, let $$w_{2t}(\boldsymbol{u})=\begin{cases}
 		1, & if\ f(\boldsymbol{u})=\max B_{p}^{f}(\boldsymbol{u},2t),\\
 		0, & otherwise,
 	\end{cases}$$
and $$Enc(\boldsymbol{u})=(\boldsymbol{u},(w_{2t}(\boldsymbol{u}))^{2t-1}),$$
where $(w_{2t}(\boldsymbol{u}))^{2t-1}$ means the $(2t-1)$-fold repetition of the bit $w_{2t}(\boldsymbol{u})$.
This gives an FCSPC for $f$ due to the following. Let $\boldsymbol{u}$, $\boldsymbol{u}^{\prime}\in\mathbb{Z}_{2}^{k}$ such that $f(\boldsymbol{u})\neq f(\boldsymbol{u}^{\prime})$. If $d_{p}(\boldsymbol{u},\boldsymbol{u}^{\prime})>2t$, then $d_{p}(Enc(\boldsymbol{u}),Enc(\boldsymbol{u}^{\prime}))\geq 2t+1.$ If $d_{p}(\boldsymbol{u},\boldsymbol{u}^{\prime})\leq 2t$, then $w_{2t}(\boldsymbol{u})\neq w_{2t}(\boldsymbol{u}^{\prime}).$
Thus, $d_{p}(Enc(\boldsymbol{u}),Enc(\boldsymbol{u}^{\prime}))\geq 2t+1.$
Hence, $r_{p}^{f}(k,t)\leq 2t-1.$
\end{proof}
According to Lemma \ref{lemma 2.4}, we give an example of pair-locally binary functions.
\begin{Example}{\rm
  Let $\mathcal{Q}=\{\boldsymbol{q}_{1},\ldots,\boldsymbol{q}_{M}\}\subseteq\mathbb{Z}_{2}^{k}$ be a code with minimum pair-distance $d_{p}=\min\limits_{i\neq j}d_{p}(\boldsymbol{q}_{i},\boldsymbol{q}_{j})$. Then the function on $\mathbb{Z}_{2}^{k}$, $$f(\boldsymbol{u})=\begin{cases}
                                                                                         i, & \mbox{if $\boldsymbol{u}=\boldsymbol{q}_{i}$,}\\
                                                                                         0, & \mbox{otherwise}
                                                                                       \end{cases}$$
  is a $\lfloor\frac{d_{p}-1}{2}\rfloor$-pair-locally function.}
\end{Example}
In order to find more pair-locally binary functions, we give the following relation between locally binary functions and pair-locally binary functions.
\begin{lem}
  If $2\leq\rho\leq k-1$ and $f\,:\,\mathbb{Z}_{2}^{k}\rightarrow Im(f)$ is a $(\rho-1)$-locally binary function, $f$ is a $\rho$-pair-locally binary function.
\end{lem}
\begin{proof}
  If $f(\boldsymbol{u}^{\prime})\in B_{p}^{f}(\boldsymbol{u},\rho)$, $d_{p}(\boldsymbol{u},\boldsymbol{u}^{\prime})\leq\rho$. Since $2\leq\rho\leq k-1$, by Lemma \ref{lemma 2.3},$$d(\boldsymbol{u},\boldsymbol{u}^{\prime})\leq d_{p}(\boldsymbol{u},\boldsymbol{u}^{\prime})-1\leq\rho-1.$$ Thus, $f(\boldsymbol{u}^{\prime})\in B_{f}(\boldsymbol{u},\rho-1)$
and $B_{p}^{f}(\boldsymbol{u},\rho)\subseteq B_{f}(\boldsymbol{u},\rho-1).$
Since $f$ is a $(\rho-1)$-locally binary function, $f$ is a $\rho$-pair-locally binary function.
\end{proof}

\section{Pair Weight Functions}
In this section, we centre around an important class of functions, pair weight functions and derive the optimal redundancy of FCSPCs designed for them. Let $f(\boldsymbol{u})=w_{p}(\boldsymbol{u})$ where $\boldsymbol{u}\in\mathbb{Z}_{2}^{k}$, $k\geq 2$. Note that $Im(w_{p})=\{0,2,\ldots,k\}$ and $E=|Im(w_{p})|=k$. For simplicity, we denote $\boldsymbol{D}_{w_{p}}^{(2)}(t,f_{1},\ldots,f_{E})$ by $\boldsymbol{D}_{w_{p}}^{(2)}(t)$.

Using Theorem \ref{theorem 3.2}, we first present an upper bound on the optimal redundancy of FCSPCs designed for the pair weight functions.
\begin{lem}\label{lemma 5.1}
  Let $f(\boldsymbol{u})=w_{p}(\boldsymbol{u})$. We have $r_{p}^{w_{p}}(k,t)\leq N_{p}(\boldsymbol{D}_{w_{p}}^{(2)}(t))$
  where $$[\boldsymbol{D}_{w_{p}}^{(2)}(t)]_{ij}=\begin{cases}
				0, & i=j,\\
				2t, & |i-j|=1,\\
				[2t+2-|i-j|]^{+}, & |i-j|>1.
			\end{cases}$$
\end{lem}
\begin{proof}
By Theorem \ref{theorem 3.2}, $r_{p}^{w_{p}}(k,t)\leq N_{p}(\boldsymbol{D}_{w_{p}}^{(2)}(t)).$
  Since $\min\limits_{\boldsymbol{u}_{1},\boldsymbol{u}_{2}\in\mathbb{Z}_{2}^{k}}d_{p}(\boldsymbol{u}_{1},\boldsymbol{u}_{2})$ s.t. $w_{p}(\boldsymbol{u}_{1})=i$, $w_{p}(\boldsymbol{u}_{2})=j$ is equal to $2$ if $|i-j|=1$ and equal to $|i-j|$ if $|i-j|>1$, $$d_{p}^{w_{p}}(i,j)=\begin{cases}
				2, & |i-j|=1,\\
			|i-j|, & |i-j|>1.
			\end{cases}$$
	Hence, $$[\boldsymbol{D}_{w_{p}}^{(2)}(t)]_{ij}=\begin{cases}
			0, & i=j,\\
			2t, & |i-j|=1,\\
			[2t+2-|i-j|]^{+}, & |i-j|>1.
		\end{cases}$$
\end{proof}

We provide a lower bound on the optimal redundancy of FCSPCs designed for the pair weight functions by using Corollary \ref{corollary 3.1}
\begin{lem}\label{lemma 5.2}
  Let $f(\boldsymbol{u})=w_{p}(\boldsymbol{u})$ and $\boldsymbol{u}_{i}=(1^{i}\,0^{k-i})$, $i\in\{0,1,\ldots,k-1\}$. We have
  $$r_{p}^{w_{p}}(k,t)\geq N_{p}(\boldsymbol{D}_{w_{p}}^{(1)}(t,\boldsymbol{u}_{0},\ldots,\boldsymbol{u}_{k-1}))$$
where $$[\boldsymbol{D}_{w_{p}}^{(1)}(t,\boldsymbol{u}_{0},\ldots,\boldsymbol{u}_{k-1})]_{ij}=\begin{cases}
		0, & i=j,\\
		[2t-|i-j|-1]^{+}, & i\neq j.
	\end{cases}$$
\end{lem}
\begin{proof}
  By Corollary \ref{corollary 3.1}, $r_{p}^{w_{p}}(k,t)\geq N_{p}(\boldsymbol{D}_{w_{p}}^{(1)}(t,\boldsymbol{u}_{0},\ldots,\boldsymbol{u}_{k-1})).$
Since $w_{p}(\boldsymbol{u}_{i})=\begin{cases}
	0, & i=0,\\
	i+1, & i\geq 1,
\end{cases}$
we see that
$$w_{p}(\boldsymbol{u}_{i})\neq w_{p}(\boldsymbol{u}_{j}),\  \text{if}\  i\neq j.$$
Because $d_{p}(\boldsymbol{u}_{i},\boldsymbol{u}_{j})=|i-j|+1$ for any $i,j\in\{0,1,\ldots,k-1\}$ with $i\neq j$, $$[\boldsymbol{D}_{w_{p}}^{(1)}(t,\boldsymbol{u}_{0},\ldots,\boldsymbol{u}_{k-1})]_{ij}=\begin{cases}
			0, & i=j,\\
			[2t-|i-j|-1]^{+}, & i\neq j.
		\end{cases}$$
\end{proof}

\begin{Remark}{\rm
In order to make the gay between the upper and lower bounds of optimal redundancy small, theoretically, we should try to find a series of words
$\boldsymbol{u}_{i}$, $i\in\{0,1,\ldots,k-1\}$ such that $$\{w_{p}(\boldsymbol{u}_{0}),w_{p}(\boldsymbol{u}_{1}),\ldots,w_{p}(\boldsymbol{u}_{k-1})\}=Im(w_{p})$$
and
$$\boldsymbol{D}_{w_{p}}^{(1)}(t,\boldsymbol{u}_{0},\ldots,\boldsymbol{u}_{k-1})=\\\boldsymbol{D}_{w_{p}}^{(1)}(t,f_{1},\ldots,f_{E}),$$
i.e.,
$d_{p}(\boldsymbol{u}_{i},\boldsymbol{u}_{j})=d_{p}^{w_{p}}(w_{p}(\boldsymbol{u}_{i}),w_{p}(\boldsymbol{u}_{j}))$ for any $i,j\in\{0,1,\ldots,k-1\}$. However, when $k\geq 4$, we find it is impossible. The following is the reason.

Let $w_{p}(\boldsymbol{u}_{1})=2$, then $\boldsymbol{u}_{1}=(0,\ldots,0,1,0,\ldots,0)$. We choose a word $\boldsymbol{u}_{2}$ such that $w_{p}(\boldsymbol{u}_{2})=3$ and $d_{p}(\boldsymbol{u}_{1},\boldsymbol{u}_{2})=d_{p}^{w_{p}}(2,3)=2$. We find that $\boldsymbol{u}_{2}=(0,\ldots,0,1,1,0,\ldots,0)$. Next, we choose a word $\boldsymbol{u}_{3}$ such that $w_{p}(\boldsymbol{u}_{3})=4$ and $d_{p}(\boldsymbol{u}_{1},\boldsymbol{u}_{3})=d_{p}^{w_{p}}(2,4)=2$. Then, $\boldsymbol{u}_{3}=(0,\ldots,0,1,0,\ldots,0,1,0,\ldots,0)$. However, $d_{p}(\boldsymbol{u}_{2},\boldsymbol{u}_{3})\geq 3>2=d_{p}^{w_{p}}(3,4)$. This contradicts to the assumption that $d_{p}(\boldsymbol{u}_{i},\boldsymbol{u}_{j})=d_{p}^{w_{p}}(w_{p}(\boldsymbol{u}_{i}),w_{p}(\boldsymbol{u}_{j}))$ for any $i,j\in\{0,1,\ldots,k-1\}$.}
\end{Remark}
To make the proofs of the following results more intuitive and understandable, we present an example.
\begin{Example}{\rm
  For $k=6$, $t=3$, one has $$\boldsymbol{D}_{w_{p}}^{(2)}(t)=\begin{pmatrix}
                                                       0 & 6 & 6 & 5 & 4 & 3 \\
                                                       6 & 0 & 6 & 6 & 5 & 4 \\
                                                       6 & 6 & 0 & 6 & 6 & 5 \\
                                                       5 & 6 & 6 & 0 & 6 & 6 \\
                                                       4 & 5 & 6 & 6 & 0 & 6 \\
                                                       3 & 4 & 5 & 6 & 6 & 0
 \end{pmatrix}$$
  and
  $$\boldsymbol{D}_{w_{p}}^{(1)}(t,\boldsymbol{u}_{0},\ldots,\boldsymbol{u}_{k-1})=\begin{pmatrix}
                                                       0 & 4 & 3 & 2 & 1 & 0 \\
                                                       4 & 0 & 4 & 3 & 2 & 1 \\
                                                       3 & 4 & 0 & 4 & 3 & 2 \\
                                                       2 & 3 & 4 & 0 & 4 & 3 \\
                                                       1 & 2 & 3 & 4 & 0 & 4 \\
                                                       0 & 1 & 2 & 3 & 4 & 0
                                                     \end{pmatrix}.$$}
\end{Example}

In the following, combining with the upper and lower bounds of $N_{p}(\boldsymbol{D})$, we provide some more simplified bounds of $r_{p}^{w_{p}}(k,t)$. Using Lemma \ref{lemma 5.2}, we get the lower bound of $N_{p}(\boldsymbol{D})$ below.
\begin{Corollary}\label{corollary 5.1}
For any $k>t$,
we have
$$r_{p}^{w_{p}}(k,t)\geq\frac{20t^{3}-20t}{9(t+1)^{2}}.$$
\end{Corollary}
\begin{proof}
  Let $\boldsymbol{u}_{i}=(1^{i}\,0^{k-i})$, $i\in\{0,1,\ldots,k-1\}$ and $\mathcal{P}=\{\boldsymbol{p}_{0},\ldots,\boldsymbol{p}_{k-1}\}$ be a $\boldsymbol{D}_{w_{p}}^{(1)}(t,\boldsymbol{u}_{0},\ldots,\boldsymbol{u}_{k-1})$-code of length $N_{p}(\boldsymbol{D}_{w_{p}}^{(1)}(t,\boldsymbol{u}_{0},\ldots,\boldsymbol{u}_{k-1}))$. Consider the first $t+1$ codewords of $\mathcal{P}$.
  Let
  $$\boldsymbol{D}_{t+1}=\boldsymbol{D}_{w_{p}}^{(1)}(t,\boldsymbol{u}_{0},\ldots,\boldsymbol{u}_{k-1})\begin{pmatrix}
                                                                                                       1 & \cdots & t+1 \\
                                                                                                       1 & \cdots & t+1
                                                                                                     \end{pmatrix}.$$
Then $\{\boldsymbol{p}_{0},\ldots,\boldsymbol{p}_{t}\}$ is a $\boldsymbol{D}_{t+1}$-code.
  By Lemma \ref{lemma 3.2}, \begin{align*}
   N_{p}(\boldsymbol{D}_{t+1}) & \geq \frac{8}{3(t+1)^{2}}\sum_{i=1}^{t+1}\sum_{j=i+1}^{t+1}[\boldsymbol{D}_{t+1}]_{ij}\\ & =\frac{8}{3(t+1)^{2}}\sum_{i=1}^{t+1}\sum_{j=i+1}^{t+1}[\boldsymbol{D}_{w_{p}}^{(1)}(t,\boldsymbol{u}_{0},\ldots,\boldsymbol{u}_{k-1})]_{ij}\\ &
  \geq \frac{8}{3(t+1)^{2}}\sum_{k=1}^{t}(2t-k-1)(t+1-k)\\ & =\frac{20t^{3}-20t}{9(t+1)^{2}}.
                            \end{align*}
Finally, using Lemma \ref{lemma 5.2}, we have
 $$r_{p}^{w_{p}}(k,t)\geq N_{p}(\boldsymbol{D}_{w_{p}}^{(1)}(t,\boldsymbol{u}_{0},\ldots,\boldsymbol{u}_{k-1})\geq N_{p}(\boldsymbol{D}_{t+1})\geq \frac{20t^{3}-20t}{9(t+1)^{2}}.$$
\end{proof}

By virtue of  Lemma \ref{lemma 5.1}, we can derive an upper bound on $r_{p}^{w_{p}}(k,t)$, as given below.
\begin{Corollary}
  For $t\geq 6$ and $2\leq k\leq (2t-1)^{2}$, $$r_{p}^{w_{p}}(k,t)\leq\frac{4t-4}{1-2\sqrt{ln(2t-1)/(2t-1)}}.$$
\end{Corollary}
\begin{proof}
From $$[\boldsymbol{D}_{w_{p}}^{(2)}(t)]_{ij}=\begin{cases}
			0, & i=j,\\
			2t, & |i-j|=1,\\
			[2t+2-|i-j|]^{+}, & |i-j|>1,
		\end{cases}$$
we see that $[\boldsymbol{D}_{w_{p}}^{(2)}(t)]_{ij}\leq 2t.$
As $t\geq 6$ and $k\leq (2t-1)^{2}$, by Corollary \ref{corollary 3.3}, $$N_{p}(\boldsymbol{D}_{w_{p}}^{(2)}(t))\leq\frac{4t-4}{1-2\sqrt{ln(2t-1)/(2t-1)}}.$$
Then by Lemma \ref{lemma 5.1}, $$r_{p}^{w_{p}}(k,t)\leq N_{p}(\boldsymbol{D}_{w_{p}}^{(2)}(t))\leq\frac{4t-4}{1-2\sqrt{ln(2t-1)/(2t-1)}}.$$
\end{proof}

With the upper and lower bounds above,  for positive integers  $k,t$ with $6\leq t<k\leq (2t-1)^{2}$,
we can narrow down the optimal redundancy between roughly $\frac{20t}{9}$ and $4t$.
By Corollary \ref{corollary 3.1}, for any function $f$ with $|Im(f)|\geq 2$, $r_{p}^{f}(k,t)\geq 2t-2$. We can also provide a tighter lower bound of the optimal redundancy of FCSPCs designed for pair-weight functions.
\begin{lem}\label{lemma 5.3}
  For $k\geq 2$, we have $r_{p}^{w_{p}}(k,t)\geq 2t-1$.
\end{lem}
\begin{proof}
  Suppose that $r=r_{p}^{w_{p}}(k,t)<2t-1$. There exists an encoding function $$Enc\,:\,\mathbb{Z}_{2}^{k}\rightarrow\mathbb{Z}_{2}^{k+r},\,\boldsymbol{u}\mapsto(\boldsymbol{u},\boldsymbol{p(u)})$$ which defines an FCSPC for the pair weight function. Let $\boldsymbol{u}=(0,\ldots,0,0)$, $\boldsymbol{u}^{\prime}=(0,\ldots,0,1)$. Then $$d_{p}(Enc(\boldsymbol{u}),Enc(\boldsymbol{u}^{\prime}))\leq 2t.$$ This contradicts the definition of FCSPC. Hence, $r_{p}^{w_{p}}(k,t)\geq 2t-1.$
\end{proof}

To get the value of the optimal redundancy more exactly, we provide a construction of FCSPCs for the pair-weight functions.
\begin{Construction}
  For $\boldsymbol{u}\in\mathbb{Z}_{2}^{k}$, we define $Enc(\boldsymbol{u})=(\boldsymbol{u},\boldsymbol{p}_{w_{p}\big(\boldsymbol{u})+1}\big)$ where the $\boldsymbol{p}_{i}'s$ are defined as follows:
Let $\boldsymbol{p}_{1},\ldots,\boldsymbol{p}_{2t+1}$ be a code with minimum pair-distance $2t$, i.e., $d_{p}(\boldsymbol{p}_{i},\boldsymbol{p}_{j})\geq 2t\ \text{for all}\ i,j\leq 2t+1,\,i\neq j.$  Set $\boldsymbol{p}_{i}=\boldsymbol{p}_{i\,\rm{smod}\,(2t+1)}$ for $i\geq 2t+2$ ($a\,\rm{smod}\,b=(a-1)\mod b+1$).
\end{Construction}
We  show that the encoding function gives an FCSPC for the pair-weight functions.
\begin{proof}
  	Let $\boldsymbol{u}_{i},\boldsymbol{u}_{j}\in\mathbb{Z}_{2}^{k}$ with $w_{p}(\boldsymbol{u}_{i})\neq w_{p}(\boldsymbol{u}_{j})$. If $|w_{p}(\boldsymbol{u}_{i})-w_{p}(\boldsymbol{u}_{j})|\geq 2t+1$,
   then
   $$d_{p}(Enc(\boldsymbol{u}_{i}),Enc(\boldsymbol{u}_{j}))\geq d_{p}(\boldsymbol{u}_{i},\boldsymbol{u}_{j})\geq |w_{p}(\boldsymbol{u}_{i})-w_{p}(\boldsymbol{u}_{j})|\geq 2t+1.$$
If $|w_{p}(\boldsymbol{u}_{i})-w_{p}(\boldsymbol{u}_{j})|\leq 2t$, then
$$d_{p}(Enc(\boldsymbol{u}_{i}),Enc(\boldsymbol{u}_{j}))\geq d_{p}(\boldsymbol{u}_{i},\boldsymbol{u}_{j})+d_{p}(\boldsymbol{p}_{i},\boldsymbol{p}_{j})-1\geq 2+2t-1=2t+1.$$
\end{proof}
Using the construction, we can derive  more exact values of the optimal redundancy.
\begin{lem}
  	For $k\geq 2$, we have $1\leq r_{p}^{w_{p}}(k,1)\leq 2$, $3\leq r_{p}^{w_{p}}(k,2)\leq 5$ and $5\leq r_{p}^{w_{p}}(k,3)\leq 7$.
\end{lem}
\begin{proof}
We first show that $1\leq r_{p}^{w_{p}}(k,1)\leq 2$. Obviously, $r_{p}^{w_{p}}(k,1)\geq 1.$
Then we only need to prove that  $r_{p}^{w_{p}}(k,1)\leq 2$ by the construction above. For $\boldsymbol{u}\in\mathbb{Z}_{2}^{k}$, we define $Enc(\boldsymbol{u})=(\boldsymbol{u},\boldsymbol{p}_{w_{p}(\boldsymbol{u})+1})$ where the $\boldsymbol{p}_{i}'s$ are defined as follows. Set $\boldsymbol{p}_{1}=(0,0)$, $\boldsymbol{p}_{2}=(1,0)$, $\boldsymbol{p}_{3}=(0,1)$ and $\boldsymbol{p}_{i}=\boldsymbol{p}_{i\,\rm{smod}\,3}$ for $i\geq 4$. Then the encoding function defines an FCSPC for the pair weight functions. Thus, $r_{p}^{w_{p}}(k,1)\leq 2.$

  Next, we prove that $3\leq r_{p}^{w_{p}}(k,2)\leq 5$. By Lemma \ref{lemma 5.3}, $r_{p}^{w_{p}}(k,2)\geq 3.$ It is enough to show  $r_{p}^{w_{p}}(k,2)\leq 5$ by the construction above. For $\boldsymbol{u}\in\mathbb{Z}_{2}^{k}$, we define $Enc(\boldsymbol{u})=(\boldsymbol{u},\boldsymbol{p}_{w_{p}(\boldsymbol{u})+1})$ where the $\boldsymbol{p}_{i}'s$ are defined as follows. Set $\boldsymbol{p}_{1}=(0,0,0,0,0)$, $\boldsymbol{p}_{2}=(1,1,1,0,0)$, $\boldsymbol{p}_{3}=(1,0,1,1,0)$, $\boldsymbol{p}_{4}=(1,1,0,0,1)$, $\boldsymbol{p}_{5}=(0,1,1,1,1)$ and $\boldsymbol{p}_{i}=\boldsymbol{p}_{i\,\rm{smod}\,5}$ for $i\geq 6$. Then the encoding function defines an FCSPC for the pair weight functions. Thus, $r_{p}^{w_{p}}(k,1)\leq 5.$

  At last, we prove that $5\leq r_{p}^{w_{p}}(k,3)\leq 7$. By Lemma \ref{lemma 5.3}, $r_{p}^{w_{p}}(k,3)\geq 5.$ Then we prove that $r_{p}^{w_{p}}(k,3)\leq 7$ by the construction above. For $\boldsymbol{u}\in\mathbb{Z}_{2}^{k}$, we define $Enc(\boldsymbol{u})=(\boldsymbol{u},\boldsymbol{p}_{w_{p}(\boldsymbol{u})+1})$ where the $\boldsymbol{p}_{i}'s$ are defined as follows. Set $\boldsymbol{p}_{1}=(0,0,0,1,1,0,0)$, $\boldsymbol{p}_{2}=(0,0,1,1,0,1,1)$, $\boldsymbol{p}_{3}=(0,1,0,0,0,1,0)$, $\boldsymbol{p}_{4}=(0,1,1,0,1,0,1)$, $\boldsymbol{p}_{5}=(1,0,0,0,1,1,1)$, $\boldsymbol{p}_{6}=(1,0,1,0,0,0,0)$, $\boldsymbol{p}_{7}=(1,1,0,1,0,0,1)$ and $\boldsymbol{p}_{i}=\boldsymbol{p}_{i\,\rm{smod}\,7}$ for $i\geq 8$. Then the encoding function defines an FCSPC for the pair weight functions. Thus, $r_{p}^{w_{p}}(k,3)\leq 7.$
\end{proof}
We note that the code $\mathcal{P}=\{(0,0,0,1,1,0,0),(0,0,1,1,0,1,1),(0,1,0,0,0,1,0),(0,1,1,0,1,0,1),\\(1,0,0,0,1,1,1),(1,0,1,0,0,0,0),(1,1,0,1,0,0,1)\}$ with minimum pair-distance $d_{p}=6$ is from  \cite[Example 8]{CL}.

\section{Pair Weight Distribution Functions}
In this section, we study FCSPCs designed for an interesting class of functions, pair-weight distribution functions. Let $f(\boldsymbol{u})=\Delta_{p}^{T}(\boldsymbol{u})\triangleq\lfloor\frac{w_{p}(\boldsymbol{u})}{T}\rfloor$ where $\boldsymbol{u}\in\mathbb{Z}_{2}^{k}$ and $k,\,T\in\mathbb{N}$ with $k\geq 2$. For simplicity, we only consider the case that $T$ divides $k+1$ and $T\geq 2t+1$. Under the condition, $|Im(f)|=E=\frac{k+1}{T}$. More precisely, the function $f$ is a step threshold function based on the pair-weight of $\boldsymbol{u}$ with $E$ steps. The function values increase by one at integer multiples of $T$. In the following, we derive the optimal redundancy of FCSPCs designed for pair-weight distribution functions.
\begin{lem}
  For any $k,t,T\in\mathbb{N}$ such that $T$ divides $k+1$ and $T\geq 2t+1$, we have $2t-2\leq r_{p}^{\Delta_{p}^{T}}(k,t)\leq 2t.$
\end{lem}
\begin{proof}
  By the Corollary \ref{corollary 3.1}, $r_{p}^{\Delta_{p}^{T}}(k,t)\geq 2t-2.$
We will show  $r_{p}^{\Delta_{p}^{T}}(k,t)\leq 2t$ by constructing an FCSPC for $\Delta_{p}^{T}$. To this end, we define $Enc_{\Delta_{p}^{T}}(\boldsymbol{u})=(\boldsymbol{u},\boldsymbol{p}_{(w_{p}(\boldsymbol{u})+1)\,\rm{smod}\,T})$ with $\boldsymbol{p}_{i}\in\mathbb{Z}_{2}^{2t}$ defined as follows. Set $\boldsymbol{p}_{i}=(1^{i-1}\,0^{2t-i+1})$ for $i\in\{1,2,\ldots,2t+1\}$ and $\boldsymbol{p}_{i}=(1^{2t})$ for $i\in\{2t+2,\ldots,T\}$.
Let $\boldsymbol{u}_{1},\boldsymbol{u}_{2}\in\mathbb{Z}_{2}^{k}$ with $f(\boldsymbol{u}_{1})\neq f(\boldsymbol{u}_{2})$. If $d_{p}(\boldsymbol{u}_{1},\boldsymbol{u}_{2})\geq 2t+1$, then  $d_{p}(Enc(\boldsymbol{u}_{1}),Enc(\boldsymbol{u}_{2}))\geq 2t+1.$
In the following, we assume that $d_{p}(\boldsymbol{u}_{1},\boldsymbol{u}_{2})\leq 2t$. Since $d_{p}(\boldsymbol{u}_{1},\boldsymbol{u}_{2})\leq 2t<T$, we let $$w_{p}(\boldsymbol{u}_{1})=(m-1)T+w_{1},\,w_{p}(\boldsymbol{u}_{2})=mT+w_{2}.$$ If $m=1$, $w_{1}\in\{0,2,\cdots,T-1\}$, $w_{2}\in\{0,1,\cdots,T-1\}.$ If $m\geq 2$, $w_{1},w_{2}\in\{0,1,\cdots,T-1\}$. Since $w_{p}(\boldsymbol{u}_{2})-w_{p}(\boldsymbol{u}_{1})\leq d_{p}(\boldsymbol{u}_{1},\boldsymbol{u}_{2})\leq 2t<T$, $T+w_{2}-w_{1}<T$. Thus, $w_{1}>w_{2}.$
By the definition, $$Enc(\boldsymbol{u}_{1})=(\boldsymbol{u}_{1},\boldsymbol{p}_{w_{1}+1})\,,\,Enc(\boldsymbol{u}_{2})=(\boldsymbol{u}_{2},\boldsymbol{p}_{w_{2}+1}).$$
If $T=2t+1$, then $$\boldsymbol{p}_{w_{1}+1}=(1^{w_{1}}\,0^{2t-w_{1}})\,,\,\boldsymbol{p}_{w_{2}+1}=(1^{w_{2}}\,0^{2t-w_{2}}).$$ If $w_{1}=2t$, then $w_{2}=0$ and $d_{p}(Enc(\boldsymbol{u}_{1}),Enc(\boldsymbol{u}_{2}))\geq 2t+1.$ Otherwise, $$d_{p}(Enc(\boldsymbol{u}_{1}),Enc(\boldsymbol{u}_{2}))\geq (w_{1}-w_{2}+1)+(T+w_{2}-w_{1})-1=T=2t+1.$$
In the case that $T>2t+1$, we have to discuss in five cases.
	\begin{itemize}
		\item If $w_{1}\leq 2t$, $w_{2}\leq 2t$, then $d_{p}(Enc(\boldsymbol{u}_{1}),Enc(\boldsymbol{u}_{2}))\geq 2t+1.$
\item If $w_{1}>2t$, $w_{2}=0$, then $d_{p}(Enc(\boldsymbol{u}_{1}),Enc(\boldsymbol{u}_{2}))\geq 2t+1.$
\item If $w_{1}>2t$, $w_{2}=2t$, then $d_{p}(Enc(\boldsymbol{u}_{1}),Enc(\boldsymbol{u}_{2}))\geq T+2t-w_{1}\geq 2t+1.$
\item If $w_{1}>2t$, $0<w_{2}<2t$, then $$d_{p}(Enc(\boldsymbol{u}_{1}),Enc(\boldsymbol{u}_{2}))\geq (2t-w_{2}+1)+(T+w_{2}-w_{1})-1=2t+T-w_{1}\geq 2t+1.$$
\item If $w_{1}>2t$, $w_{2}>2t$, then  $$d_{p}(Enc(\boldsymbol{u}_{1}),Enc(\boldsymbol{u}_{2}))\geq T+w_{2}-w_{1}\geq T-[(T-1)-(2t+1)]=2t+2>2t+1.$$
\end{itemize}
Hence, we conclude that $r_{p}^{\Delta_{p}^{T}}(k,t)\leq 2t.$
\end{proof}

\begin{table}[h]
  \centering
  \begin{tabular}{ccc}
  \hline
  $w_{p}(\boldsymbol{u})$ & $\boldsymbol{p}_{(w_{p}(\boldsymbol{u})+1)smodT}$ & $f(\boldsymbol{u})$ \\
  \hline
  $0$ & $\boldsymbol{p}_{1}$ & \multirow{4}*{$0$} \\
  $2$ & $\boldsymbol{p}_{3}$ & ~ \\
  $\vdots$ & $\vdots$ & ~ \\
  $T-1$ & $\boldsymbol{p}_{T}$ & ~ \\
  \hline
  $T$ & $\boldsymbol{p}_{1}$ & \multirow{4}*{$1$} \\
  $T+1$ & $\boldsymbol{p}_{2}$ & ~ \\
  $\vdots$ & $\vdots$ & ~ \\
  $2T-1$ & $\boldsymbol{p}_{T}$ & ~ \\
  \hline
  $\vdots$ & $\vdots$ & $\vdots$ \\
  \hline
  $(\frac{k+1}{T}-1)T$ & $\boldsymbol{p}_{1}$ & \multirow{4}*{$\frac{k+1}{T}-1$} \\
  $(\frac{k+1}{T}-1)T+1$ & $\boldsymbol{p}_{2}$ & ~ \\
  $\vdots$ & $\vdots$ & ~ \\
  $k$ & $\boldsymbol{p}_{T}$ & ~ \\
  \hline
\end{tabular}
  \caption{Parity Vectors Assigned to Words with Different Pair-weights.}
\end{table}

\section{Comparison to the Redundancy of  Symbol-Pair Codes}
At the beginning of the paper, we claim that using function-correcting symbol-pair codes reduces the redundancy when the message is long and the image of the function is small. In this section, we show that the redundancy is really reduced when the function belong to the three classes of functions discussed above. We use the notations as in Figures \ref{Figure1} and  \ref{Figure2}.

We first estimate the redundancy of classical symbol-pair codes (CSPCs), $n-k$. By the Sphere Packing bound, $2^{n}\geq 2^{k}B(n,t).$
Thus, $n-k\geq \log_{2}B(n,t).$ According to the computational formula of $B(n,t)$, one has $B(n,t)=\Theta(n^{\lfloor\frac{t}{2}\rfloor}).$
Thus, we can roughly get $n-k\geq \lfloor\frac{t}{2}\rfloor \log_{2}n.$ By the Singleton bound, $2^{k}\leq 2^{n-(2t+1)+2}.$
We see that  $n\geq k+2t+1.$
Hence, $$n-k\geq \lfloor\frac{t}{2}\rfloor \log_{2}(k+2t-1).$$

In the following, we compare the redundancy of CSPCs with that of FCSPCs by Table \ref{Table2}.
\begin{table}[h]
  \centering
  \begin{tabular}{ccc}
    \hline
    Functions & CSPC & FCSPC \\
    \hline
    2t-pair-locally binary functions & $\lfloor\frac{t}{2}\rfloor\log_{2}(k+2t-1)$ & $2t-1$  \\
    Pair weight functions & $\lfloor\frac{t}{2}\rfloor\log_{2}(k+2t-1)$ & $\frac{4t-4}{1-2\sqrt{ln(2t-1)/(2t-1)}}$  \\
    Pair weight distribution functions & $\lfloor\frac{t}{2}\rfloor\log_{2}(k+2t-1)$ & $2t$ \\
    \hline
  \end{tabular}
  \caption{The Comparison between the Redundancy of CSPCs and that of FCSPCs.}\label{Table2}
\end{table}

In Table \ref{Table2}, the contents of column 2 are the lower bounds of redundancy of CSPCs
and the contents of column 3 are the upper bounds of optimal redundancy of FCSPCs which the redundancy of FCSPCs can always achieve.
Comparing the column 2 and the column 3, we find that the redundancy is reduced provided $k$ is large enough.
For example, when $k>17-2t$, $\lfloor\frac{t}{2}\rfloor\log_{2}(k+2t-1)>2t-1$.
Thus, the redundancy can be reduced by using FCSPCs.
We note that this is just a very rough estimate with some constraints of $k$ and $t$,
but it can reflect the key advantage of FCSPCs.

\section{Conclusion}
In this paper, we extend the notion of  function correcting codes from binary symmetric channels to   symbol-pair read channels.
The new paradigm meets the need of protecting a certain attribute of the message against errors
when the message is read through the symbol-pair read channels and improves the efficiency of information storage and reading by reducing the redundancy.
Our key goal is to construct FCSPCs and derive the optimal redundancy of FCSPCs.
By connecting FCSPCs with irregular-pair-distance codes, we provide some upper and lower bounds of the optimal redundancy for generic functions and also do research on some specific functions.

Due to the difficulty of constructing FCSPCs explicitly, we are always unable to get the exact value of optimal redundancy or confirm
whether these bounds we present are tight or not.
So constructing more FCSPCs with small redundancy is an important further research direction.
In addition, due to the connection between FCSPCs and irregular-pair-distance codes, finding more effective ways to compute or estimate $N_{p}(\boldsymbol{D})$ is also significant for determining the optimal redundancy or optimizing the bounds we present.

Apart from these directions, we can also consider function correcting codes for other functions of interest or under other channels.

\vskip 2mm
{\bf Acknowledgements}~ This work was supported by National Natural Science Foundation of China under Grant
Nos. 12271199, 12171191, 12371521 and The Fundamental Research Funds for the Central Universities 30106220482.

\end{document}